\documentclass[a4paper,twocolumn,10pt,accepted=2022-04-05]{quantumarticle}
\pdfoutput=1
\usepackage[utf8]{inputenc}
\usepackage[english]{babel}
\usepackage[T1]{fontenc}
\usepackage[numbers]{natbib}
\usepackage{amsthm,amsmath,amssymb}
\usepackage{hyperref}
\usepackage{cleveref}

\usepackage{xspace}

\usepackage{booktabs}

\usepackage{tikz}
\usepackage{pgfplots}
\usepgfplotslibrary{statistics}
\usepackage[braket, qm]{qcircuit}

\usepackage[multiple]{footmisc}

\usepackage{todonotes}

\newcommand{\co}{\colon\thinspace}
\newcommand{\eps}{\varepsilon}
\newcommand{\<}{\langle}
\renewcommand{\>}{\rangle}

\DeclareMathOperator{\tr}{tr}
\DeclareMathOperator{\vol}{vol}

\newcommand{\A}{\mathfrak{A}}

\newcommand{\C}{\mathcal{C}}
\newcommand{\CX}{\mathit{CX}}
\newcommand{\CAN}{\mathit{CAN}}
\newcommand{\Haar}{\mathrm{Haar}}
\newcommand{\I}{\mathit{I}}
\newcommand{\lrs}{\texttt{lrs}\xspace}

\newcommand{\PU}{\mathit{PU}}
\newcommand{\pu}{\mathfrak{pu}}
\newcommand{\R}{\mathbb{R}}
\renewcommand{\S}{\mathcal{S}}

\newcommand{\su}{\mathfrak{su}}
\newcommand{\SWAP}{\mathit{SWAP}}
\renewcommand{\th}{\textsuperscript{th}\xspace}
\newcommand{\XX}{\mathit{XX}}
\newcommand{\XY}{\mathit{XY}}
\newcommand{\X}{\mathit{X}}
\newcommand{\YX}{\mathit{YX}}
\newcommand{\YY}{\mathit{YY}}
\newcommand{\Y}{\mathit{Y}}
\newcommand{\ZZ}{\mathit{ZZ}}
\newcommand{\IZ}{\mathit{IZ}}
\newcommand{\ZI}{\mathit{ZI}}
\newcommand{\Z}{\mathit{Z}}

\theoremstyle{plain}
\newtheorem{theorem}{Theorem}[section]
\newtheorem{lemma}[theorem]{Lemma}
\newtheorem{corollary}[theorem]{Corollary}
\theoremstyle{definition}
\newtheorem{definition}[theorem]{Definition}

\newtheorem{procedure}[theorem]{Procedure}
\newtheorem*{question*}{Question}

\newtheorem*{importantremark*}{Important Remark}
\theoremstyle{remark}
\newtheorem{remark}[theorem]{Remark}
\newtheorem{example}[theorem]{Example}

\begin{document}

\title{Optimal synthesis into fixed $\XX$ interactions}
\author{Eric C.\ Peterson}
\email{peterson.eric.c@gmail.com}
\affiliation{IBM Quantum, San Jose, CA, USA}
\author{Lev S.\ Bishop}
\email{lsbishop@us.ibm.com}
\affiliation{IBM Quantum, Yorktown Heights, NY, USA}
\author{Ali Javadi-Abhari}
\email{Ali.Javadi@ibm.com}
\affiliation{IBM Quantum, Yorktown Heights, NY, USA}

\maketitle

\begin{abstract}
    We describe an optimal procedure, as well as its efficient software implementation, for exact and approximate synthesis of two-qubit unitary operations into any prescribed discrete family of $\XX$--type interactions and local gates.
    This arises from the analysis and manipulation of certain polyhedral subsets of the space of canonical gates.
    Using this, we analyze which small sets of $\XX$--type interactions cause the greatest improvement in expected infidelity under experimentally-motivated error models.
    For the exact circuit synthesis of Haar-randomly selected two-qubit operations, we find an improvement in estimated infidelity by $\approx$31.4\% when including $\CX^{1/2}$ and $\CX^{1/3}$ alongside the standard gate $\CX$, near to the optimal limit of $\approx$36.9\% obtained by including all fractional applications $\CX^\alpha$, $\alpha \in [0, 1]$.
\end{abstract}


\section{Introduction}

In this paper, we describe an optimal synthesis routine for two-qubit unitary operations which targets any discrete family of two-qubit gates, each locally equivalent to some $\exp(-ia\XX)$.
We refer to such gates as being of \emph{$\XX$--type} and note that this class includes all controlled unitaries.
Gate sets of $\XX$--type are common on contemporary platforms: the gate $\CX$ is an example, and synthesis routines for it have long been known to give rise to algorithmic schemes for universal quantum computation, making it an attractive target for device engineers.
The physical processes which give rise to the operation $\CX$ can typically be truncated to produce ``fractional applications'' $\CX^{\alpha}$ for $0 \le \alpha \le 1$, each of which is \emph{also} of $\XX$--type, giving rise to an infinite family of further examples.%
\footnote{
Up to local equivalence, this family is an exhaustive set of examples.
}
Such fractional applications can be found on devices based on superconducting qubits (e.g., IBM's), as well as on those based on ion traps.
Though not required for universal computation, the availability of these ``overcomplete'' basis gates has the potential to yield more efficient synthesized circuits, particularly if the error magnitude of $\CX^{\alpha}$ correlates with $\alpha$: while the universally programmable $\CX$ circuit invokes $\CX$ three times, the universally programmable fractional $\CX$ circuit invokes $\CX^\alpha$, $\CX^\beta$, and $\CX^\gamma$ with $\<\alpha + \beta + \gamma\> = 3/2$.%
\footnote{%
See \cite{BMBRAH} for numerical justification or \cite{Mehta} for a proof.
}

In practice, however, these parametric families are difficult to operate.
The relationship between the degree of physical process truncation and the value $\alpha$ is often nonlinear and prone to imperfect measurement, and constraints in the steering electronics (e.g., waveform sample rate) can make truncations unavailable below some threshold, so that wholesale use of these families may be prohibited on realistic hardware.
Still, precision can be guaranteed for any particular value of $\alpha$, which gives rise to the following question:

\begin{question*}
Given a fixed ``calibration budget'' which permits the tuning of $n$ fractional operations, which set of values $\alpha_1$, \ldots, $\alpha_n$ maximizes (average-case) device performance?
How does one efficiently find expressions for generic two-qubit unitaries in terms of these operations?
How does one simultaneously guarantee the optimality of such expressions, as measured against device performance?
\end{question*}

We answer this question fully.
Our fundamental results are an efficient test for when a two-qubit unitary operation admits expression as a circuit using any particular sequence of $\XX$--interaction strengths $(\alpha_1, \ldots, \alpha_n)$ with local gates interleaved (\Cref{MainGlobalTheorem}), an efficicent synthesis routine for manufacturing such circuits (\Cref{EffectiveSynthesis}, \Cref{MainLocalTheorem}), and an efficient routine for producing the best approximation (in average gate fidelity) within the set of such circuits (\Cref{SynthesisWrapper}).%
\footnote{%
To our knowledge, this is the first optimal synthesis algorithm which targets a ``heterogeneous'' set of multiqubit interactions.
}
These tools combine to give an optimal synthesis scheme for reasonably behaved cost functions (e.g., average gate infidelity).
An implementation of our technique can be found in Qiskit's \texttt{quantum\_info} subpackage as the class \texttt{XXDecomposer}~\cite{Qiskit}, where it can be verified that it outperforms blind numerical search in both wall time and output quality (cf.~\Cref{SynthesisWallTimeFig}).
We leverage these results to explore the design space of gate set extensions where we constrain $\alpha_1, \ldots, \alpha_n$ to be drawn from some small, fixed set of pretuned angles.
For experimentally realistic error models,%
\footnote{
In practice, we find that these models amount to a \emph{weighted} count of circuit elements, where the weights depend linearly on the exponents $\alpha_j$.
}
our main findings are first that including first $\CX^{1/2}$ and then $\CX^{1/3}$ give significant improvement over $\CX$ alone at several common tasks%
\footnote{
Specifically, we examine synthesis of random unitaries, as encountered during whole-circuit resynthesis, and reported on in \Cref{NiceAnglesTable}; and we exaxmine synthesis of certain structured operators like QFT, as encountered in peephole optimization of highly structured circuits, and reported on in \Cref{QFTStatsTable}.
},
and second that these two gates capture almost all of the benefit of allowing $\alpha_1, \ldots, \alpha_n$ to be drawn without constraint (cf.\ \Cref{OptimaTable}).
Finally, we note that some of our proofs rely on using a computer algebra system to manipulate polytopes, and we have released both this framework and the proof software under the Qiskit umbrella of packages as \href{https://github.com/Qiskit/monodromy}{\texttt{monodromy}}~\cite{monodromy}.

\begin{figure}
    \centering
    \scalebox{0.6}{$U = \left(\begin{array}{cccc}
     0.098+0.121j & -0.26 -0.391j & -0.443-0.322j & -0.662+0.133j \\
     0.748-0.132j &  0.372-0.084j & -0.432-0.027j &  0.261-0.148j \\
    -0.137-0.365j & -0.453+0.334j & -0.28 -0.547j &  0.296-0.257j \\
    -0.474-0.147j &  0.44 -0.352j & -0.265-0.255j &  0.328+0.439j
    \end{array}\right)$}
    \vspace{\baselineskip}
    
    \scalebox{0.7}{
    \Qcircuit @C=1.0em @R=0.2em @!R { \\
    	 	& \gate{L} & \ctrl{1} & \gate{L} & \ctrl{1} & \gate{L} & \ctrl{1} & \gate{L} & \qw \\ 
    	 	& \gate{L} & \targ & \gate{L} & \targ & \gate{L} & \targ & \gate{L} & \qw \\ 
    \\ }}
    \scalebox{0.7}{
    \Qcircuit @C=1.0em @R=0.2em @!R { \\
    	 	& \gate{L} & \multigate{1}{\XX_{\frac{\pi}{8}}} & \gate{L} & \multigate{1}{\XX_{\frac{\pi}{8}}} & \gate{L} & \multigate{1}{\XX_{\frac{\pi}{8}}} & \gate{L} & \qw \\ 
    	 	& \gate{L} & \ghost{\XX_{\frac{\pi}{4}}} & \gate{L} & \ghost{\XX_{\frac{\pi}{8}}} & \gate{L} & \ghost{\XX_{\frac{\pi}{8}}} & \gate{L} & \qw \\ 
    \\ }}
    \scalebox{0.7}{
    \Qcircuit @C=1.0em @R=0.2em @!R { \\
    	 	& \gate{L} & \multigate{1}{\XX_{\frac{\pi}{12}}} & \gate{L} & \multigate{1}{\XX_{\frac{\pi}{12}}} & \gate{L} & \multigate{1}{\XX_{\frac{\pi}{12}}} & \gate{L} & \multigate{1}{\XX_{\frac{\pi}{12}}} & \gate{L} & \qw \\ 
    	 	& \gate{L} & \ghost{\XX_{\frac{\pi}{6}}} & \gate{L} & \ghost{\XX_{\frac{\pi}{12}}} & \gate{L} & \ghost{\XX_{\frac{\pi}{12}}} & \gate{L} & \ghost{\XX_{\frac{\pi}{12}}} & \gate{L} & \qw \\ 
    \\ }}
    \scalebox{0.7}{
    \Qcircuit @C=1.0em @R=0.2em @!R { \\
    	 	& \gate{L} & \multigate{1}{\XX_{\frac{\pi}{12}}} & \gate{L} & \multigate{1}{\XX_{\frac{\pi}{12}}} & \gate{L} & \multigate{1}{\XX_{\frac{\pi}{8}}} & \gate{L} & \qw \\ 
    	 	& \gate{L} & \ghost{\XX_{\frac{\pi}{12}}} & \gate{L} & \ghost{\XX_{\frac{\pi}{12}}} & \gate{L} & \ghost{\XX_{\frac{\pi}{8}}} & \gate{L} & \qw \\ 
    \\ }}
    
    \vspace{\baselineskip}
    \begin{tabular}{@{}lrrr@{}} \toprule
    & \multicolumn{3}{c}{Relative infidelity} \\ \cmidrule(r){2-4}
    Gateset description & \cite{ZVSWControlled} & \cite{LMMB} & Ours \\ \midrule
    $\{\XX_{\frac{\pi}{4}}\}$ (baseline) & 133\% & \textbf{100\%} & \textbf{100\%} \\
    $\{\XX_{\frac{\pi}{8}}\}$ & 67\% & 67\% & \textbf{50\%} \\
    $\{\XX_{\frac{\pi}{12}}\}$ & 56\% & 56\% & \textbf{44\%} \\
    $\{\XX_{\frac{\pi}{4}}, \XX_{\frac{\pi}{8}}, \XX_{\frac{\pi}{12}}\}$ & n/a & n/a & \textbf{39\%} \\ \bottomrule
    \end{tabular}
    \caption{%
    Four syntheses of a two-qubit operator $U \in SU(4)$ with canonical coordinate $(0.968, 0.273, 0.038)$.
    \textbf{(1)}
    An exact, optimal synthesis into a triple of $\CX$ gates.
    \textbf{(2)} 
    An exact, optimal synthesis into a triple of $\XX_{\frac{\pi}{8}}$ gates.
    \textbf{(3)} 
    An exact, optimal synthesis into four $\XX_{\frac{\pi}{12}}$ gates.
    \textbf{(4)}
    An exact, optimal synthesis into a mixed set of gates.
    Finally, we include the relative infidelity costs of these syntheses in an error model where $\XX$ gate infidelity is linearly related to the parameter with a small affine offset.
    \textbf{Relative infidelity:}
    We assume that $\XX$ gate infidelity is linearly related to the parameter with a small affine offset and that infidelity is approximately additive over that of its gates.
    We refer to the infidelity of circuit (1) as the ``baseline'', and we report as ``relative infidelity'' the percentage of this value achieved by other synthesis strategies.
    \emph{See \Cref{SynthesisWallTimeFig} and \Cref{GatesetOptSection} for similar statistics where $U$ is allowed to range.}
    }
    \label{ExampleCircuits}
\end{figure}

\begin{figure}
    \centering
    \begin{tabular}{@{}ccccc@{}} \toprule
    & \multicolumn{2}{c}{Gates emitted} \\ \cmidrule(r){2-3}
    $n$ & $\CX$ & $\{\XX_{\frac{\pi}{4}}, \XX_{\frac{\pi}{8}}, \XX_{\frac{\pi}{12}}\}$ & Rel.\ Infidelity \\
    \midrule
    3 & $\CX^{6}$ & $\XX_{\pi/12}^{2} \cdot \XX_{\pi/8}^{2}$ & 27.8\% \\
    4 & $\CX^{12}$ & $\XX_{\pi/12}^{6} \cdot \XX_{\pi/8}^{3}$ & 29.2\% \\
    5 & $\CX^{20}$ & $\XX_{\pi/12}^{12} \cdot \XX_{\pi/8}^{4}$ & 30.0\% \\
    6 & $\CX^{30}$ & $\XX_{\pi/12}^{20} \cdot \XX_{\pi/8}^{5}$ & 30.6\% \\
    7 & $\CX^{42}$ & $\XX_{\pi/12}^{30} \cdot \XX_{\pi/8}^{6}$ & 31.0\% \\
    \bottomrule
    \end{tabular}
    \caption{%
    Qiskit syntheses of QFT circuits over $n$ qubits, targeting a family of qubits supporting either $\S = \{\CX\}$ or $\S = \{\XX_{\frac{\pi}{4}}, \XX_{\frac{\pi}{8}}, \XX_{\frac{\pi}{12}}\}$ and with all-to-all connectivity.
    At right, we include the expected circuit infidelity, reported as a fraction of that of traditional synthesis methods, under the assumptions that $\XX$ gate infidelity is linearly related to the parameter with a small affine offset and that infidelity is approximately additive.
    In the limit of a large qubit count, the expected infidelity of a QFT circuit synthesized to the fractional gate set drops by two-thirds that of the standard $\CX$--based gate set.
    We don't intend circuit synthesis for QFTs to be a ``killer app'', but rather as evidence that these methods are not limited to random inputs.
    }
    \label{QFTStatsTable}
\end{figure}

\subsection*{Related literature}
We give a non-exhaustive survey of existing ideas.

\begin{description}
    \item[\cite{SMB}]
    Shende, Markov, and Bullock showed how to synthesize optimal-depth circuits for $\CX$--based gate sets.
    \item[\cite{CBSNG}]
    Cross et al.\ extended that to synthesize optimal-infidelity circuits for $\CX$--based gate sets.
    \item[\cite{ZVSWControlled,ZYG,YZG}]
    Zhang, Vala, Sastry, and Whaley (see also Zhang, Ye, and Guo) gave a synthesis method for gate sets based on controlled unitary operations.
    Our method constructs circuits out of the same building blocks, but our method for how to arrange those blocks is different.
    This difference in synthesis strategy gives us \emph{optimal} circuits, whereas often\footnote{More exactly, for $\frac{1}{2} (5 + 4 \log(2) - 6 \log(3)) \approx$ 59\% of possible native $\XX$ interactions; cf.\ \cite[Figure 2]{ZVSWControlled}.} they miss by a constant offset.
    \item[\cite{PCS}]
    Peterson et al.\ showed how to detect when a two-qubit unitary operation admits expression using a given circuit type with sufficiently many freely ranging local operations.
    This gives rise to a method for analyzing the optimality of a synthesis strategy, but it does not show how to perform the actual synthesis.
    \item[\cite{ETE}]
    Earnest, Tornow, and Egger have shown how to produce the entire family of $\XX$--type interactions from a particular pulse-level implementation of $\CX$, which then permits the use of a straightforward synthesis method.
    Their method does not extend to other hardware implementations of $\CX$, including the implementation used by the IBM group to achieve quantum volume 64~\cite{JJBLBBCGIKKKKLLMNNNPRSSWWWYZDCG}, which does not lend itself to non-calibrated scaling.
    \item[\cite{HDWKZNSZC}]
    Huang et al.\ have recently shown how to synthesize (optimal) circuits for a gateset containing $\sqrt{\mathit{ISWAP}}$, a particular ``fractional $\mathit{ISWAP}$''.
    Despite surface similarities, our results depart substantially: $\sqrt{\mathit{ISWAP}}$ is not of $\XX$--type; they work with the fixed gate $\sqrt{\mathit{ISWAP}}$ as opposed to an unknown family of fractional applications; and they consider circuits of depth at most $3$.
    \item[\cite{LMMB}]
    Lao et al.\ have used brute-force numerical search to uncover which fractional interactions might be valuable to include in a native gate set put to various specific purposes.
    This has substantial overlap with our discussion of gate set optimization, but it does not solve the synthesis problem: numerical search is both non-optimal and two orders of magnitude slower than our direct synthesis.
\end{description}

\subsection*{The case of one qubit}

To give a sense of our methods and results, let us analyze the analogous problem for one-qubit unitaries: decomposition into a fixed set of fractional $\X$--rotations and unconstrained $\Z$--rotations.
The fixed $\X$--rotation most typically available is $\X_{\frac{\pi}{2}}$, and an $\X_{\frac{\pi}{2}}$--based circuit can be synthesized for a unitary $U$ through ``Euler $\Z\Y\Z$ decomposition''.
Namely, there are angle values $\theta$, $\phi$, and $\lambda$ which satisfy \[U = \Z_\phi \cdot \Y_\theta \cdot \Z_\lambda = \Z_{\phi + \pi} \cdot \X_{\frac{\pi}{2}} \cdot \Z_{\theta + \pi} \cdot \X_{\frac{\pi}{2}} \cdot \Z_\lambda,\] easily calculable by diagonalizing $UU^T$.
Since $U$ can freely range, the right-hand side of this equation gives a universally programmable quantum circuit.
A downside to this circuit is that the operational cost of $U$ is always that of a pair of $\X_{\frac{\pi}{2}}$ gates, even if $U$ itself is a ``small'' rotation of the Bloch sphere.%
\footnote{%
One can set aside special cases when $\theta$ is zero or $\frac{\pi}{2}$, but these are probability-zero events in common measures.
}


For circuits based on other choices of choices of $\X$--rotation angles, such as \[Y_\theta = \Z_{\iota'} \cdot \X_\psi \cdot \Z_\iota \cdot \X_{\psi'} \cdot \Z_{\iota''},\] one perform some mathematical analysis to discern the limited set of synthesizable operations $\Y_\theta$, ultimately arriving at the critical relationship
\[\cos \theta = \cos \psi \cdot \cos \psi' - \cos \iota \cdot \sin \psi \cdot \sin \psi',\]
and the remaining parameters $\iota'$, $\iota''$ can be explicitly determined by inspecting complex phases.
Varying $\iota$, $\iota'$, and $\iota''$, this trigonometric equation admits a solution precisely when $\theta$ satisfies \[|\psi - \psi'| \le \theta \le \pi - \left| \pi - (\psi + \psi') \right|.\]
Let us refer to this interval as $I_{\psi, \psi'}$.
In the same manner, a longer sequence of interactions $\X_{\psi_1}, \ldots, \X_{\psi_n}$ interleaved with $\Z$--rotations gives rise to a corresponding interval $I_{\psi_1, \ldots, \psi_n}$ of achievable values of $\theta$.

Suppose that any given gate $\X_\psi$, with $\psi \in [0, \pi]$, can be made available in an experimental setting with infidelity \[m \cdot \left( \frac{\pi}{2}\right)^{-1} \cdot \psi + b\] for some error model parameters $m$ and $b$, and at a fixed calibration cost per gate.
We seek a small set of gates $\{\X_{\psi_1}, \ldots, \X_{\psi_n}\}$ so that the intervals constructed above cover the possible values of $\theta$ so as to minimize the expected infidelity cost of a given operation.
For instance, \Cref{WeylRegionsFor1Q} shows the relevant intervals for the gate set $\{\X_{\frac{\pi}{2}}, \X_{\frac{\pi}{3}}\}$.
Several aspects of this goal can also be understood with additional nuance:
\begin{description}
    \item[``Expected'':]
    The distribution of operations $U$ to be synthesized will affect the relative importance of the various choices of $\psi$.
    A safe assumption is that $U$ is drawn according to the Haar distribution, in which case the distribution of angle values $\theta$ is given by $p(\theta) d\theta = \frac{1}{2} \sin(\theta) d\theta$.
    \item[``Cost'':]
    In addition to the operational cost of a synthesized circuit (i.e., the cost from gate applications), one can also incorporate a cost stemming from synthesizing some $\theta'$ rather than the requested $\theta$.
    There are then some circumstances where it is profitable to \emph{deliberately} mis-synthesize $\Y_\theta$ as $\Y_{\theta'}$, provided the difference between $\theta$ and $\theta'$ is small and the difference in operational cost between the two circuits is large.
    Average gate infidelity gives a popular embodiment of this idea, where the fidelity of two one-qubit operations is given by the formula \[F_{\mathrm{avg}}(\Y_\theta, \Y_{\theta'}) = \frac{2 + \cos(\theta - \theta')}{3}.\]
    Against this yardstick, the $\theta' \in I_{\psi, \psi'}$ which gives the best approximation to a $\theta \not\in I_{\psi, \psi'}$ occurs at one of the interval endpoints.
    \item[``Given operation'':]
    Rather than synthesizing the operation $U$ requested, the compiler can choose to inject a reversible logic operation $R$ and its inverse $R^{-1}$ into the program, synthesizing the composite $U \cdot R$ and either commuting $R^{-1}$ forward through the circuit or absorbing its effect into software.
    This option can be used to further shape the expected distribution of inputs.
    For single-qubit operations, a typical choice of $R$ is the classical logic gate $\X_\pi$, which has the effect of trading $\X_\theta$ for $\X_{\pi - \theta}$.%
    \footnote{%
    For two-qubit operations, a common choice of $R$ is the classical logic gate $\SWAP$.
    }
\end{description}

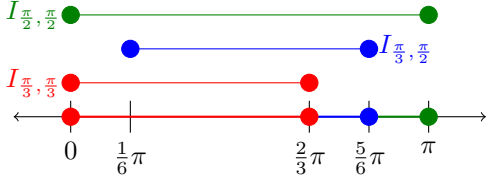
\begin{figure}
    \centering
    \begin{tikzpicture}[scale=1.5]
    \draw[<->] (-0.5,0) -- (3.6415,0);
    
    \draw (0pt, 4pt) -- (0pt, -4pt) node[below] {$0$};
    \draw[shift={(1/6*3.14, 0)}, color=black] (0, 4pt) -- (0, -4pt) node[below] {$\frac{1}{6} \pi$};
    \draw[shift={(2/3*3.14, 0)}, color=black] (0, 4pt) -- (0, -4pt) node[below] {$\frac{2}{3} \pi$};
    \draw[shift={(5/6*3.14, 0)}, color=black] (0, 4pt) -- (0, -4pt) node[below] {$\frac{5}{6} \pi$};
    \draw[shift={(3.14, 0)}, color=black] (0, 4pt) -- (0, -4pt) node[below] {$\pi$};
    
    \draw[red,fill=red] (0,0.3) circle (.5ex);
    \draw[red,fill=red] (2/3 * 3.14,0.3) circle (.5ex);
    \draw[-, red] (0, 0.3) node[left] {$I_{\frac{\pi}{3}, \frac{\pi}{3}}$} -- (2/3 * 3.14, 0.3);
    
    \draw[blue,fill=blue] (1/6 * 3.14, 0.6) circle (.5ex);
    \draw[blue,fill=blue] (5/6 * 3.14, 0.6) circle (.5ex);
    \draw[-, blue] (1/6 * 3.14, 0.6) -- (5/6 * 3.14, 0.6) node[right] {$I_{\frac{\pi}{3}, \frac{\pi}{2}}$};
    
    \draw[green!50!black,fill=green!50!black] (0 * 3.14, 0.9) circle (.5ex);
    \draw[green!50!black,fill=green!50!black] (1 * 3.14, 0.9) circle (.5ex);
    \draw[-, green!50!black] (0 * 3.14, 0.9) node[left] {$I_{\frac{\pi}{2}, \frac{\pi}{2}}$} -- (3.14, 0.9);
    
    \draw[green!50!black,fill=green!50!black] (1 * 3.14,0) circle (.5ex);
    \draw[-, thick, green!50!black] (5/6 * 3.14, 0.0) -- (3.14, 0.0);
    \draw[blue,fill=blue] (5/6 * 3.14,0) circle (.5ex);
    \draw[-, thick, blue] (2/3 * 3.14, 0.0) -- (5/6 * 3.14, 0.0) ;
    \draw[-, thick, red] (0, 0.0) -- (2/3 * 3.14, 0.0);
    \draw[red,fill=red] (0,0) circle (.5ex);
    \draw[red,fill=red] (2/3 * 3.14,0) circle (.5ex);
    \end{tikzpicture}
    \caption{%
    The optimal synthesis regions for the one-qubit gate set $\{\X_{\frac{\pi}{2}}, \X_{\frac{\pi}{3}}, \Z_{\mathit{cts}}\}$.
    The interval being covered is the set of angles $[0, \pi]$ appearing as the middle parameter in a $\mathit{ZYZ}$--decomposition of a generic $U \in PU(2)$.}
    \label{WeylRegionsFor1Q}
\end{figure}

\noindent
Considering only exact synthesis for now, we compute the following expected (i.e., Haar-averaged) average gate infidelities for various gate sets:
\begin{description}
    \item[$\X_{\frac{\pi}{2}}$:] $2m + 2b$, the standard decomposition, used as a baseline.
    \item[$\X_{\frac{\pi}{3}}$:] $\frac{3}{2}m + \frac{9}{4} b$, an improvement of $\approx 25\%$ over the baseline, provided $b \ll m$.
    \item[$\{\X_{\frac{\pi}{2}}, \X_{\frac{\pi}{3}}\}$:] $\frac{19 - \sqrt{3}}{12} m + 2b$, an improvement of $\approx 28\%$ over the baseline, provided $b \ll m$.
    \item[$\{\X_{\mathit{cts}}\}$:] In the continuous limit with all gates $\X_\theta$ available, the cost becomes $m + b$, an improvement of $50\%$ over the baseline.
\end{description}
These values can be further improved by considering approximate synthesis, mirrored synthesis, or both.

\subsection*{Outline}

Our analysis of the two-qubit case follows along the same lines as above, and in the same order.
\begin{description}
    \item[{\Cref{KAKResume}}:]
    Generalizing Euler decomposition, we give a lightning review of $KAK$ decomposition as specialized to two-qubit unitary operations.
    \item[{\Cref{SynthesisSection}}:]
    We describe a more detailed plan of attack on the two-qubit problem, outlining the steps in the proofs to come.
    \item[{\Cref{GlobalTheoremSection}}:]
    Generalizing the interval $I_{\psi, \psi'}$, we give a compact description of which two-qubit gates are accessible to a circuit built out of a fixed sequence of $\XX$--type interactions with one-qubit operations interleaved (\Cref{MainGlobalTheorem}).
    This leverages previous work of Peterson et al.~\cite{PCS}: it detects when a two-qubit operation admits synthesis as a circuit of a certain type, but it does not indicate how to produce the circuit.
    \item[{\Cref{LocalTheoremSection}}:]
    Generalizing the formula relating $\cos \theta$ and $\cos \iota$, we single out a method for choosing local circuit parameters which are simple to analyze (\Cref{InterferenceInequalities}).
    We then compare with \Cref{GlobalTheoremSection} and prove that each of these restricted circuits nonetheless exhaust the space of possibilities (\Cref{MainLocalTheorem}).
    \item[{\Cref{OptimalSynthesisSection}}:]
    Generalizing the discussion around cost, we provide an efficient method to find the best approximation within a given circuit family (\Cref{ApproximationTheorem}), and we couple it to the preceding results to produce the promised efficient synthesis method (\Cref{EffectiveSynthesis}, \Cref{SynthesisWrapper}).
    \item[{\Cref{GatesetOptSection}}:]
    Generalizing the calculations of expected cost for various gate sets, we use experimental data to justify the use of a particular error model as a cost function (\Cref{AffineErrorModelDefn}), study the effect of choice of gate set (\Cref{Cost2DExample}), and describe what is left to gain in the large limit of fractional gate count (\Cref{EfficiencyLowerBound}).
\end{description}

\noindent
We give a small example of the effectiveness of these techniques as applied to a random operator in \Cref{ExampleCircuits} and to a family of structured operators in \Cref{QFTStatsTable}, reserving further analysis for \Cref{GatesetOptSection}.

\subsection*{Conventions}

We use the following abbreviations throughout:
\begin{align*}
    c_\theta & = \cos(\theta), &
    s_\theta & = \sin(\theta), &
    x_+ & = \sum_j x_j.
\end{align*}

\section{R\'esum\'e on two-qubit unitaries and the monodromy map}\label{KAKResume}

We briefly recall the theory of Cartan decompositions as it applies to two-qubit unitary operations and its role in circuit synthesis.

\begin{lemma}[{\cite{KrausCirac}, \cite{Makhlin}, \cite{ZVSWGeometric}, \cite{PCS}}]\label{TraditionalWeylDecompTheorem}
Let $\CAN$ denote the following two-qubit gate: \[\CAN(a_1, a_2, a_3) = \exp(-i(a_1 \XX + a_2 \YY + a_3 \ZZ)) = \]
\scalebox{0.8}{$\left(
\begin{array}{cccc}
e^{i a_3} c_{a_1 - a_2} & 0 & 0 & -i e^{i a_3} s_{a_1 - a_2} \\
0 & e^{-i a_3} c_{a_1 + a_2} & -i e^{-i a_3} s_{a_1 + a_2} & 0 \\
0 & -i e^{-i a_3} s_{a_1 + a_2} & e^{-i a_3} c_{a_1 + a_2} & 0 \\
-i e^{i a_3} s_{a_1 - a_2} & 0 & 0 & e^{i a_3} c_{a_1 - a_2}
\end{array}
\right).$}
Any two-qubit unitary operation $U \in \PU(4)$ can be written as \[U = L \cdot \CAN(a_1, a_2, a_3) \cdot L',\] where $L, L' \in \PU(2)^{\times 2}$ are local gates and $a_1, a_2, a_3$ are (underdetermined) real parameters.
\qed
\end{lemma}

\begin{definition}[{``Canonical decomposition'', cf.\ \cite{PCS}}]\label{CanonicalDecomp}
In \Cref{TraditionalWeylDecompTheorem}, there is a unique triple $(a_1, a_2, a_3)$ satisfying $a_1 \ge a_2 \ge a_3 \ge 0$, $\frac{\pi}{2} \ge a_1 + a_2$, and one of either $a_3 > 0$ or $a_1 \le \frac{\pi}{4}$.
Such a triple is called a \textit{positive canonical coordinate}, and we denote the space of such as $\A_{C_2}$.
This unicity determines a function $\Pi\co PU(4) \to \A_{C_2}$, called the \textit{monodromy map}.
Away from the plane $a_3 = 0$, this function is continuous.%
\footnote{Near the plane $a_3 = 0$, the function $\Pi$ becomes continuous after imposing the identification $(a_1, a_2, 0) \sim (\frac{\pi}{2} - a_1, a_2, 0)$.}
\end{definition}

\begin{example}
Here are the positive canonical triples for some familiar gates:
\begin{align*}
    \Pi(I) & = (0, 0, 0), \\
    \Pi(\SWAP) & = \left(\frac{\pi}{4}, \frac{\pi}{4}, \frac{\pi}{4}\right), \\
    \Pi(\CX) & = \left(\frac{\pi}{4}, 0, 0\right).
\end{align*}
Generalizing the last example, the positive canonical triple for any controlled unitary gate has the form $(a_1, 0, 0)$; we say that such an operation is \textit{of $\XX$--type}, and we abbreviate such gates to \[\XX_{a_1} = \CAN(a_1, 0, 0).\]
Specifically, the fractional gate $\CX^{\alpha}$ is of $\XX$--type, with coordinate $\Pi(\CX^\alpha) = (\alpha \cdot \frac{\pi}{4}, 0, 0)$, so that \[\CX^\alpha = (H \otimes I) \cdot \XX_{\alpha \cdot \frac{\pi}{4}} \cdot (H \cdot \Z_{-\frac{\pi}{2}} \otimes \Z_{-\frac{\pi}{2}}) \equiv \XX_{\alpha \cdot \frac{\pi}{4}}.\]
From this perspective, the varying coordinate measures interaction duration or interaction strength, so that smaller values give rise to less entanglement.
\end{example}

For us this apparatus has two main uses, captured in the following pair of results:

\begin{lemma}
A pair of two-qubit operations $U$ and $V$ are said to be \textit{locally equivalent} when there exist local gates $L, L' \in \PU(2)^{\times 2}$ with $U = L \cdot V \cdot L'$.
This condition holds if and only if $\Pi(U) = \Pi(V)$.
\qed
\end{lemma}

\begin{theorem}[{\cite{PCS}; \cite{monodromy}}]\label{PCSPolytopeTheorem}
Let $\S, \S' \subseteq PU(4)$ be two sets of two-qubit operations whose images $\Pi(\S), \Pi(\S') \subseteq \A_{C_2}$ through $\Pi$ are polytopes (e.g., a set of isolated points).
The image of the set \[\S \cdot \PU(2)^{\times 2} \cdot \S' = \left\{S \cdot (L \otimes L') \cdot S' \middle| \begin{array}{c} S \in \S, S' \in \S', \\ L, L' \in \PU(2) \end{array} \right\}\] through $\Pi$ is then also a polytope.
Given explicit descriptions of the input polytopes as families of linear inequalities, the output polytope can also be so described.
\qed
\end{theorem}

Our work will lead us directly into considering families of two-qubit gates and their parametrizations, so we introduce some attendant language.

\begin{definition}
A \textit{gate set} is any collection of two-qubit unitaries; typically we will consider gate sets which are made up of finitely many controlled unitaries.
For a gate set $\S$, an $\S$--circuit is a finite sequence of members of $\S$ and local gates.
The operation which it enacts is given by the product of the sequence elements.
A \textit{circuit shape}%
\footnote{Also commonly called a \textit{(circuit) template}.}
is a circuit-valued function \[C\co \theta \mapsto (L_0(\theta), S_1, \ldots, L_{n-1}(\theta), S_n, L_n(\theta)),\] where each $L_j$ is a parametrized local operator and each $S_j \in \S$ is fixed.

It can be convenient to place further restrictions on the $L_j$ (e.g., that they consist only of $\Z$--rotations), but absent explicit mention we take each $L_j$ to surject onto $\PU(2)^{\times 2}$.
In this surjective case, the sequence $(S_1, \ldots, S_n)$ determines the image of $C$, and it follows from \Cref{PCSPolytopeTheorem} that the image of $\Pi \circ C$ in $\A_{C_2}$ is a polytope, called the \textit{circuit polytope} of $C$ (or of $(S_1, \ldots, S_n)$).
In the case that $\S$ consists of gates of $\XX$--type, locally surjective circuits can be further identified with the underlying sequence of interaction strengths $(\alpha_1, \ldots, \alpha_n)$ with $S_j \equiv \XX_{\alpha_j}$.
\end{definition}

\begin{remark}
The coordinate system given in \Cref{CanonicalDecomp} is not unique: a similar theorem holds for any choice of ``Weyl alcove'' in $\pu_4$.
When $\mathfrak g$ is the Lie algebra of a simply connected Lie group (e.g., $\su_4$), each Weyl alcove is related to every other by a discrete set of linear transformations including reflections and shears.
Without the simply-connected hypothesis (e.g., $\pu_4$), they are related by linear transformations and ``scissors congruence''.
Our choice of coordinate system differs from that used in previous work of Peteron et al.~\cite{PCS} by a nontrivial scissors congruence, effectively replacing the condition stated there, \[(\operatorname{LogSpec} U)_3 + 1/2 > (\operatorname{LogSpec} U)_1,\] by the alternative \[(\operatorname{LogSpec} U)_2 + (\operatorname{LogSpec} U)_3 > 0.\]
\end{remark}

\begin{remark}[{cf.\ \Cref{DisorderedMainTheorem}}]\label{EarlyDisorderedRemark}
Later, it will be convenient for us to consider a variant of positive canonical triples which are not required to be sorted.
Unsorted triples $(a_1, a_2, a_3)$ which become positive canonical triples upon sorting are those which satisfy $0 \le a_j$ and $a_j + a_k \le \frac{\pi}{2}$ for all choices of $j$ and $k$.
The set of such triples still gives a convex polytope.
\end{remark}

\section{$\XX$-based synthesis: Strategy}\label{SynthesisSection}

Let us turn to the task of synthesizing for any two-qubit unitary operator $U$ and gate set $\S$ consisting of $\XX$ interactions an $\S$--circuit $C$ modeling $U$.
To set the stage for our strategy, suppose instead that $C$ is given.
We can produce from it a sequence of truncations $C_j$ that retain steps $1$ through $j$.
Each $C_j$ is also a circuit modeling some other unitary operator $U_j$, and if $C$ is optimal for circuits modeling $U$ against some well-behaved cost function (e.g., operation count), then each $C_j$ will be so optimal for $U_j$.
The images $p_j = \Pi(U_j) \in \A_{C_2}$ of these intermediate operators then describe a path through the Weyl alcove, where the $j$\th step in the path belongs to the region $P_j$ of operations whose optimal circuits take the form of $C_j$.

Since our goal is to construct $C$, we might instead begin by constructing the path $(p_j)_j$, subject to the two constraints:
\begin{enumerate}
    \item $p_j$ lies in $P_j$.
    \item The hop from $p_j$ to $p_{j+1}$ is given by some nice circuit.
\end{enumerate}

In order to understand the first constraint, we give a compact description of $P_j$ by way of describing the circuit polytope associated to an arbitrary sequence of $\XX$--type interactions.
We call this the \textit{global theorem}~(\Cref{MainGlobalTheorem}) since it describes the large-scale structure of the problem and does not reference the individual point $p_j$.
Though our main tool here is \Cref{PCSPolytopeTheorem}, for a generic sequence of interactions it can only guarantee an exponential-sized family of convex bodies, themselves each of increasing facet complexity.
It is a special feature of interactions of $\XX$--type that the associated circuit polytopes have a fixed number of convex bodies, each of fixed complexity, independent of the sequence length.

To understand the second constraint, we choose a particular ``nice circuit'' and analyze the effect under $\Pi$ of appending such a circuit to a canonical gate~(\Cref{XXYYProductsOpaque}), resulting in a family of constraints we call ``interference inequalities''~(\Cref{InterferenceInequalities}).
This, too, is specific to our case: even for interactions of $\XX$--type, not all choices of unit circuit have a discernable image under $\Pi$, nevermind a polytope.

We complete the program by linking these two together in the \textit{local theorem}~(\Cref{MainLocalTheorem}): we show that for any $p_{j+1} \in P_{j+1}$, we can always find a $p_j \in P_j$ linked to $p_{j+1}$ by one of these simple circuits.
This argument can then be reorganized into a constructive, efficient synthesis routine~(\Cref{EffectiveSynthesis}).
Additionally, we show how to select a point $p' \in P_j$ which is the best approximation by average gate infidelity to $p = \Pi(U)$~(\Cref{ApproximationTheorem}).

\section{The global theorem}\label{GlobalTheoremSection}

One of our overall goals is to describe the set of positive canonical triples whose optimal circuit implementation uses a sequence of interaction strengths $(\alpha_1, \ldots, \alpha_n)$.
This can be accomplished by describing those positive canonical triples which admit \emph{any} such circuit implementation, even if suboptimal.
Optimality can then be enforced by taking a complement against positive canonical triples which admit superior circuit implementations.
In this section, we accomplish this goal, summarized in the following Theorem:

\begin{theorem}\label{MainGlobalTheorem}
Let $(\alpha_j \in [0, \frac{\pi}{4}])_j$ be a sequence of interaction strengths, and let $(a_1, a_2, a_3)$ be a positive canonical coordinate.
The canonical operator $\CAN(a_1, a_2, a_3)$ admits a presentation as a circuit of the form \[L_0 \cdot \XX_{\alpha_1} \cdot L_1 \cdot \cdots \cdot L_{n-1} \cdot \XX_{\alpha_n} \cdot L_n,\] where $L_j$ are local operators, if and only if either of the following two families of linear inequalities is satisfied:
\begin{align*}
\left\{\begin{array}{rclr}
\alpha_+ & \ge & a_1 + a_2 + a_3, \\
\min_k \alpha_+ - 2\alpha_k & \ge & -a_1 + a_2 + a_3, \\
\min_{k \ne \ell} \alpha_+ - \alpha_k - \alpha_\ell & \ge & a_3;
\end{array} \right. \\
\left\{\begin{array}{rclr}
-\frac{\pi}{2} + \alpha_+ & \ge & -a_1 + a_2 + a_3, \\
\frac{\pi}{2} + \min_k \alpha_+ -2 \alpha_k & \ge & a_1 + a_2 + a_3, \\
\min_{k \ne \ell} \alpha_+ - \alpha_k - \alpha_\ell & \ge & a_3.
\end{array} \right.
\end{align*}
We respectively refer to the first, second, and third inequalities in each family as the strength, slant, and frustrum bounds.
\end{theorem}

\begin{importantremark*}
From a physical perspective, the circuit polytope ought to be invariant under injecting extra zero-strength interactions into the defining sequence of interaction strengths.
Accordingly, we always treat expressions like ``$\min_{k \ne \ell} \alpha_+ - \alpha_k - \alpha_\ell$'' as if the sequence were padded by arbitrarily many zero entries.
\end{importantremark*}

\begin{proof}[Proof of \Cref{MainGlobalTheorem}]
For the base case, note that the empty list of interaction strengths yields the polytope \[a_1 = a_2 = a_3 = 0,\] which agrees with the set of circuits locally equivalent to the identity interaction.

Suppose then that we have established the claim for a sequence of interaction strengths $(\alpha_1, \ldots, \alpha_n)$, and we would like to establish the claim for $(\alpha_1, \ldots, \alpha_n, \beta)$ for some new interaction strength $\beta$.
By allowing the $(n+1)$ different strengths to range, we note the region in the claim is naturally expressed as a polytope in $(n+1) + 3$ dimensions.
In fact, we can reduce it to a certain $6$--dimensional polytope as follows: writing $\alpha'$ and $\alpha''$ respectively for the largest and second-largest elements in the hypothesized sequence of interaction strengths, we may rewrite the inequality families above as
\begin{align*}
\left\{\begin{array}{rclr}
\alpha_+ & \ge & a_1 + a_2 + a_3, \\
\alpha_+ - 2\alpha' & \ge & -a_1 + a_2 + a_3, \\
\alpha_+ - \alpha' - \alpha'' & \ge & a_3;
\end{array} \right. \\
\left\{\begin{array}{rclr}
-\frac{\pi}{2} + \alpha_+ & \ge & -a_1 + a_2 + a_3, \\
\frac{\pi}{2} + \alpha_+ - 2\alpha' & \ge & a_1 + a_2 + a_3, \\
\alpha_+ - \alpha' - \alpha'' & \ge & a_3.
\end{array} \right.
\end{align*}
with the additional constraints
\begin{align*}
    n \cdot \frac{\pi}{4} & \ge \alpha_+, &
    \alpha_+ & \ge \alpha' + \alpha'', &
    \frac{\pi}{4} \ge \alpha' & \ge \alpha'' \ge 0.
\end{align*}
Altogether, these statements over $a_1, a_2, a_3, \alpha_+, \alpha', \alpha''$ describe a pair of convex polytopes in $6$--dimensional space.

\Cref{PCSPolytopeTheorem} gives an explicit, finite family of linear inequalities (the ``monodromy polytope'') so that $a_1, a_2, a_3, a_1', a_2', a_3', b_1, b_2, b_3$ satisfies the constraints if and only if there is a local gate $L$ and a local equivalence \[\CAN(a_1, a_2, a_3) \cdot L \cdot \CAN(a'_1, a'_2, a'_3) \equiv \CAN(b_1, b_2, b_3).\]
We combine this with the polytope from the inductive hypothesis so that its coordinates are shared with $(a_1, a_2, a_3)$ and the coordinates $(a'_1, a'_2, a'_3)$ are set to $(\beta, 0, 0)$.
This produces a union of convex polytopes in $10$--dimensional space, a point of which simultaneously captures:
\begin{description}
    \item[$(\alpha_+, \alpha', \alpha'')$:] Values extracted from the prefix of interaction strengths.
    \item[$(a_1, a_2, a_3)$:] A positive canonical coordinate which admits expression as an $\XX$--circuit as in the inductive hypothesis.
    \item[$\beta$:] A new interaction strength.
    \item[$(b_1, b_2, b_3)$:] A canonical coordinate which admits expression as a concatenation of the aforementioned circuit, a local gate, and $\XX_\beta$.
\end{description}

Our goal is to describe a certain \emph{projection} of this polytope.
Projection has the effect of introducing an existential quantifier into the above description: a point belongs to the projection of a polytope exactly when it is possible to extend the projected point by the discarded coordinates so that it satisfies the original constraints.
This trades the actual data housed in the lost coordinates---which may be complicated to the point of distraction---for the mere predicate that such data exists.
In our case, we seek to project away the coordinates $(a_1, a_2, a_3)$, which leaves only constraints on $(b_1, b_2, b_3)$, given in terms of $(\alpha_+, \alpha', \alpha'', \beta)$, ensuring that a prefix circuit of the indicated type \emph{exists}, without actually naming it.

To compute this projection, we apply Fourier--Motzkin elimination to project away the remaining coordinates and eliminate redundancies in the resulting inequality set.
These reduced inequality sets have the following form:
\begin{align*}
\left\{\begin{array}{rclr}
\alpha_+ + \beta & \ge & b_1 + b_2 + b_3, \\
\min \left\{ \begin{array}{c} \alpha_+ + \beta - 2\alpha' \\ \alpha_+ + \beta - 2\beta \end{array} \right\} & \ge & -b_1 + b_2 + b_3, \\
\min\left\{ \begin{array}{c} \alpha_+ + \beta - \alpha' - \alpha'' \\ \alpha_+ + \beta - \alpha' - \beta \\ \alpha_+ + \beta - \beta - \alpha'' \end{array} \right\} & \ge & b_3, \\
\end{array} \right. \\
\left\{\begin{array}{rclr}
-\frac{\pi}{2} + \alpha_+ + \beta & \ge & -b_1 + b_2 + b_3; \\
\min\left\{ \begin{array}{c} \frac{\pi}{2} + \alpha_+ + \beta - 2\alpha' \\ \frac{\pi}{2} + \alpha_+ + \beta - 2\beta \end{array}\right\} & \ge & b_1 + b_2 + b_3, \\
\min\left\{ \begin{array}{c} \alpha_+ + \beta - \alpha' - \alpha'' \\ \alpha_+ + \beta - \alpha' - \beta \\ \alpha_+ + \beta - \beta - \alpha'' \end{array}\right\} & \ge & b_3,
\end{array} \right.
\end{align*}
where we have collected the inequalities which give communal upper bounds into single expressions using ``$\min$''.
Notationally absorbing $\beta$ into the sequence of interaction strengths completes the proof.

See \texttt{check\_main\_xx\_theorem} in \texttt{monodromy}~\cite{monodromy} for an executable proof.
\end{proof}

\begin{example}
\begin{figure}
    \centering
    \includegraphics[width=0.4\textwidth, trim={3cm 0cm 6cm 6cm}, clip]{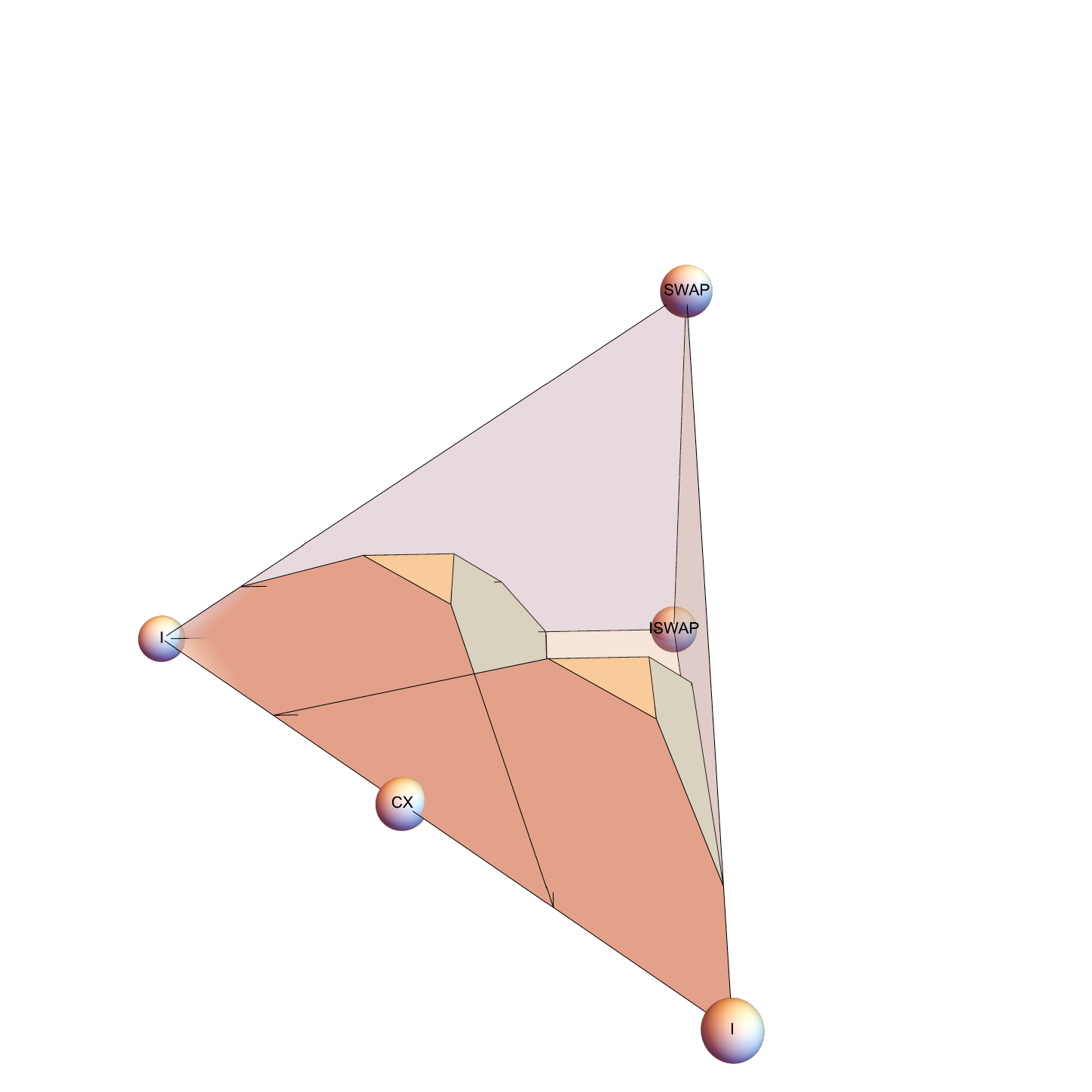}
    \includegraphics[width=0.4\textwidth, trim={5cm 4.5cm 3.5cm 8.5cm}, clip]{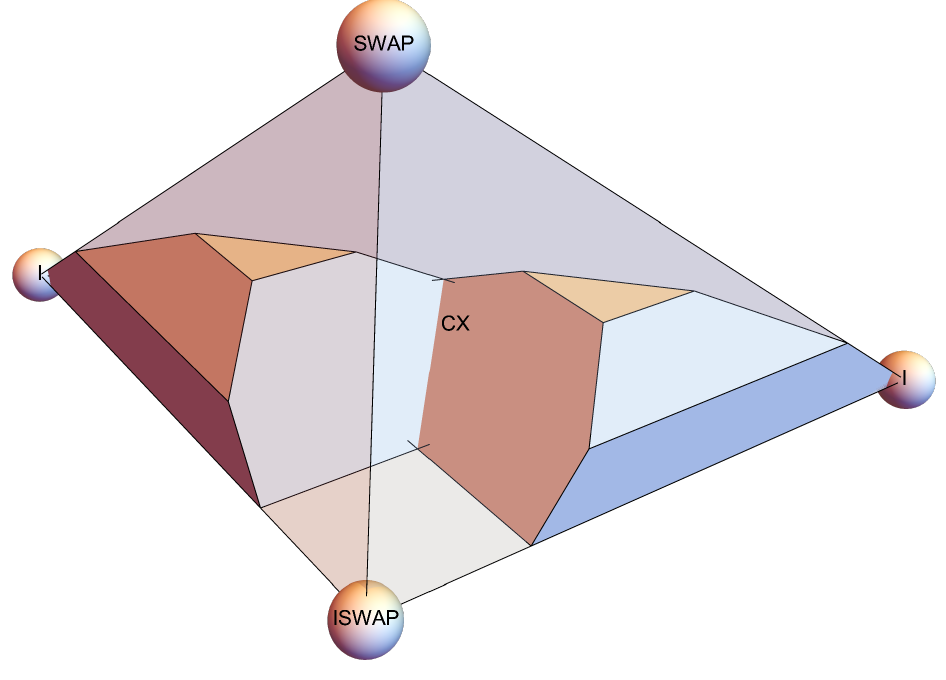}
    \caption{%
    Two perspectives on the $\XX$--circuit polytope for the interaction strength sequence $(\frac{\pi}{6}, \frac{\pi}{8}, \frac{\pi}{12})$.
    From the perspective where $\SWAP$ is near the eyepoint, the strength faces are the pair facing inwards, the slant faces are the pair facing outwards, and the frustrum face is colored tan.
    }
    \label{XXPolytopeFigure}
\end{figure}
We include a visualization of an example $\XX$--circuit polytope in \Cref{XXPolytopeFigure}.
\end{example}

\begin{remark}
The two convex bodies in the statement of \Cref{MainGlobalTheorem} are related by the linear transformation \[(a_1, a_2, a_3) \mapsto \left(\frac{\pi}{2} - a_1, a_2, a_3\right).\]
\end{remark}

\begin{remark}
\Cref{MainGlobalTheorem} is manifestly invariant under permutation of the interaction strengths.
\end{remark}

\begin{remark}\label{DisorderedMainTheorem}
By dropping the assumption that the entries of positive canonical triples are ordered descending (as in \Cref{EarlyDisorderedRemark}), we can rewrite the above inequality families in a manner that is more pleasingly symmetric.
For example, the first\footnote{The second is similar, but less pleasing to the eye.} family is rewritten as:
\begin{align*}
\alpha_+ & \ge a_+, \\
\min_k \alpha_+ - 2\alpha_k & \ge \min_k a_+ - 2a_k, \\
\min_{k \ne \ell} \alpha_+ - \alpha_k - \alpha_\ell & \ge \min_{k \ne \ell} a_+ - a_k - a_\ell.
\end{align*}
We note that we have won these pleasing formulas by losing convexity: the ``$\min$''s appearing in the lower bounds encode disjunctions of linear sentences rather than conjunctions, so we see merely a non-convex union of these convex polytopes.
\end{remark}

\begin{remark}[{\cite{ZVSWControlled}}]
In the case of a uniform interaction strength $\alpha$, we compute the quantities appearing in the upper bounds of \Cref{MainGlobalTheorem} to be
\begin{align*}
    \alpha_+ & = n \alpha, \\
    \min_k \alpha_+ - 2 \alpha_k & = (n-2) \alpha, \\
    \min_{k \ne \ell} \alpha_+ - \alpha_k - \alpha_\ell & = (n - 2)\alpha.
\end{align*}
This causes the slant and frustrum inequalities to degenerate, which recovers a theorem of Zhang et al.\ as a special case.
\end{remark}

\section{The local theorem}\label{LocalTheoremSection}

In this section, we study the problem of appending a single new $\XX$ interaction strength $\beta$ to a \emph{specific circuit} formed from a sequence of strengths $(\alpha_1, \ldots, \alpha_n)$.
Note that \Cref{MainGlobalTheorem} gives us an understanding of the ``global'' effect of appending $\XX_\beta$, where the interaction strengths are fixed but the circuit is allowed to range.
Note also that if we are able to achieve such a local understanding, we would then like to use it in reverse: given a point $p_{n+1}$ which \Cref{MainGlobalTheorem} guarantees to be modelable using a circuit with strengths $(\alpha_1, \ldots, \alpha_n, \beta)$, we would like to guarantee the existence of---and algorithmically identify!---a point $p_n$ which is modelable by $(\alpha_1, \ldots, \alpha_n)$ and for which $p_{n+1}$ is reachable by appending $\XX_\beta$ and some local gates.

Excepting the caveat about algorithmic identification, this can be accomplished directly using nothing more than the methods of the monodromy polytope.
However, because we are interested in circuit construction, we restrict what sorts of circuits we are willing to append to those of the particularly simple form given in \Cref{XXYYProductsOpaque}.
In trade, the method of the monodromy polytope no longer directly applies.

We show in \Cref{InterferenceInequalities} the ``forward'' direction of the strategy described above, then in \Cref{MainLocalTheorem} the ``reverse'' direction, culminating in the recursive step in a synthesis procedure whose full description we defer to \Cref{OptimalSynthesisSection}.
First, however, we introduce the simplified circuit which we will consider.

\begin{lemma}\label{XXYYProductsOpaque}
For any choice of $a_1$, $a_2$, $\beta$, $d$, $e$ in \[U := \CAN(a_1, a_2, 0) \cdot (\Z_d \otimes \Z_e) \cdot (\XX_\beta),\] there exist values $r$, $s$, $t$, $u$, $b_1$, and $b_2$ so that the operator $U$ may be equivalently expressed as \[(\Z_r \otimes \Z_s) \cdot \CAN(b_1, b_2, 0) \cdot (\Z_t \otimes \Z_u).\]
\end{lemma}
\begin{proof}
The vector subspace \[\mathfrak g = \operatorname{span} \{\XX, \YY, \XY, \YX, \ZI, \IZ\}\] forms a Lie subalgebra of $\pu_4$, and
the subspaces
\begin{align*}
    \mathfrak a & = \operatorname{span}\{\XX, \YY\}, &
    \mathfrak k & = \operatorname{span}\{\IZ, \ZI\}
\end{align*}
give rise to a $KAK$ decomposition yielding the desired result.
\end{proof}

Next, we note that this choice of simple local gates gives rise to the desired explicit expressions for the gate parameters.

\begin{lemma}\label{InterferenceExpressions}
Except for the outer parameters $r$, $s$, $t$, and $u$, the parameters in \Cref{XXYYProductsOpaque} are related by the equations
\begin{align*}
\frac{c_{a_1 - a_2}^2 c_\beta^2 + s_{a_1 - a_2}^2 s_\beta^2 - c_{b_1 - b_2}^2}{2 c_{a_1 - a_2} c_\beta s_{a_1 - a_2} s_\beta} & = c_{2(d+e)}, \\
\frac{c_{a_1 + a_2}^2 c_\beta^2 + s_{a_1 + a_2}^2 s_\beta^2 - c_{b_1 + b_2}^2}{2 c_{a_1 + a_2} c_\beta s_{a_1 + a_2} s_\beta} & = c_{2(d-e)}.
\end{align*}
The outer parameters $r$, $s$, $t$, and $u$ can then be deduced from a linear system with input the phases of the top half of the left-hand matrix.
\end{lemma}
\begin{proof}
The trigonometric equalities follow by equating the square-norms of the matrix entries in \Cref{XXYYProductsOpaque}.
The $(1, 1)$ and $(3, 3)$ entries respectively yield
\begin{align*}
    |c_{b_1 - b_2}|^2 & = |c_\beta c_{a_1 - a_2} - e^{2i(d+e)} s_\beta s_{a_1 - a_2}|^2, \\
    |c_{b_1 + b_2}|^2 & = |c_\beta c_{a_1 + a_2} - e^{-2i(d-e)} s_\beta s_{a_1 + a_2}|^2,
\end{align*}
where we have used the absolute values to suppress some of the phases.
We then apply the identity \[|x + r e^{i \theta}|^2 = x^2 + r^2 + 2 x r \cos \theta\] and isolate $d$ and $e$ to deduce the statement.

The linear system then arises by inspecting the phases of any nondegenerate quadruple of entries.
For example, the nonzero entries in the top half, read left-to-right, have respective phases
\begin{align*}
\exp(-i(r+s+t+u)), \\
\exp(-i(r-s+t+u), \\
-i \exp(-i(-r+s-t+u)), \\
-i \exp(-i(r+s-t-u)).
\end{align*}
This collection of linear combinations is of full rank.
\end{proof}

\noindent
We can interpret the constraints imposed by these expressions on the positive canonical triples in terms of $\beta$.

\begin{theorem}[``Interference inequalities'']\label{InterferenceInequalities}
For positive canonical triples $(a_1, a_2, 0)$ and $(b_1, b_2, 0)$ and for $\beta$ an interaction strength, there exist parameters $d$ and $e$ satisfying \[\CAN(a_1, a_2, 0) \cdot (\Z_d \otimes \Z_e) \cdot \XX_\beta \equiv \CAN(b_1, b_2, 0)\] if and only if the following inequalities hold:%
\footnote{%
It is extremely unusual that image under $\Pi$ of a circuit with constrained local gates is again a polytope.
This is, perhaps, the most important ingredient in our approach.
}
\begin{align*}
a_1 + a_2 - \beta & \le b_1 + b_2 \le \frac{\pi}{2} - \left| \frac{\pi}{2} - (a_1 + a_2 + \beta) \right|, \\
| a_1 - a_2 - \beta | & \le b_1 - b_2 \le a_1 - a_2 + \beta.
\end{align*}
Moreover, the local gates witnessing the equivalence can be taken to be $\Z$--rotations.
\end{theorem}
\begin{proof}
Starting from \Cref{InterferenceExpressions}, there exist solutions to $d$ and $e$ exactly when the following inequalities are met:
\begin{align*}
|c_{a_1 - a_2}^2 c_\beta^2 + s_{a_1 - a_2}^2 s_\beta^2 - c_{b_1 - b_2}^2| & \le |2 c_{a_1 - a_2} c_\beta s_{a_1 - a_2} s_\beta|, \\
|c_{a_1 + a_2}^2 c_\beta^2 + s_{a_1 + a_2}^2 s_\beta^2 - c_{b_1 + b_2}^2| & \le |2 c_{a_1 + a_2} c_\beta s_{a_1 + a_2} s_\beta|.
\end{align*}
Using the inequalities $a_1 + a_2 \le \frac{\pi}{2}$, $a_1 \ge a_2$, and $0 \le \beta \le \frac{\pi}{4}$, we see that both the right-hand quantities are always positive, hence the right-hand absolute value can be suppressed.
The left-hand absolute value can be equivalently expressed as a pair of inequalities, giving
\[-2 c_{a_1 - a_2} c_\beta s_{a_1 - a_2} s_\beta \le c_{a_1 - a_2}^2 c_\beta^2 + s_{a_1 - a_2}^2 s_\beta^2 - c_{b_1 - b_2}^2,\]
\[c_{a_1 - a_2}^2 c_\beta^2 + s_{a_1 - a_2}^2 s_\beta^2 - c_{b_1 - b_2}^2 \le 2 c_{a_1 - a_2} c_\beta s_{a_1 - a_2} s_\beta,\]
\[-2 c_{a_1 + a_2} c_\beta s_{a_1 + a_2} s_\beta \le c_{a_1 + a_2}^2 c_\beta^2 + s_{a_1 + a_2}^2 s_\beta^2 - c_{b_1 + b_2}^2,\]
\[c_{a_1 + a_2}^2 c_\beta^2 + s_{a_1 + a_2}^2 s_\beta^2 - c_{b_1 + b_2}^2 \le 2 c_{a_1 + a_2} c_\beta s_{a_1 + a_2} s_\beta.\]
Factoring the quadratics then yields
\[(c_{a_1 - a_2} c_\beta - s_{a_1 - a_2} s_\beta)^2 \le c_{b_1 - b_2}^2,\]
\[c_{b_1 - b_2}^2 \le (c_{a_1 - a_2} c_\beta + s_{a_1 - a_2} s_\beta)^2,\]
\[(c_{a_1 + a_2} c_\beta - s_{a_1 + a_2} s_\beta)^2 \le c_{b_1 + b_2}^2,\]
\[c_{b_1 + b_2}^2 \le (c_{a_1 + a_2} c_\beta + s_{a_1 + a_2} s_\beta)^2.\]
Rewriting the binomials as cosines of differences / sums and then converting square cosines to double-angle cosines yields
\begin{align*}
c_{2(a_1 - a_2 + \beta)} & \le c_{2(b_1 - b_2)} \le c_{2(a_1 - a_2 - \beta)}, \\
c_{2(a_1 + a_2 + \beta)} & \le c_{2(b_1 + b_2)} \le c_{2(a_1 + a_2 - \beta)}.
\end{align*}
Finally, we use the piecewise monotonicity and reflection invariance of cosine, as well as the bounds on the inputs, to deduce inequalities on the angles:
\[2(a_1 - a_2 + \beta) \ge 2(b_1 - b_2),\]
\[2(b_1 - b_2) \ge \max\{2(a_1 - a_2 - \beta), -2(a_1 - a_2 - \beta)\},\]
\[\min\{2(a_1 + a_2 + \beta), 2\pi - 2(a_1 + a_2 + \beta)\} \ge 2(b_1 + b_2).\]
\[2(b_1 + b_2) \ge 2(a_1 + b_2 - \beta).\]
Linear rearrangement yields the claimed inequality family.
\end{proof}


\begin{example}
\begin{figure}
    \centering
    \includegraphics[width=0.4\textwidth, trim={2cm 0.5cm 11cm 4cm}, clip]{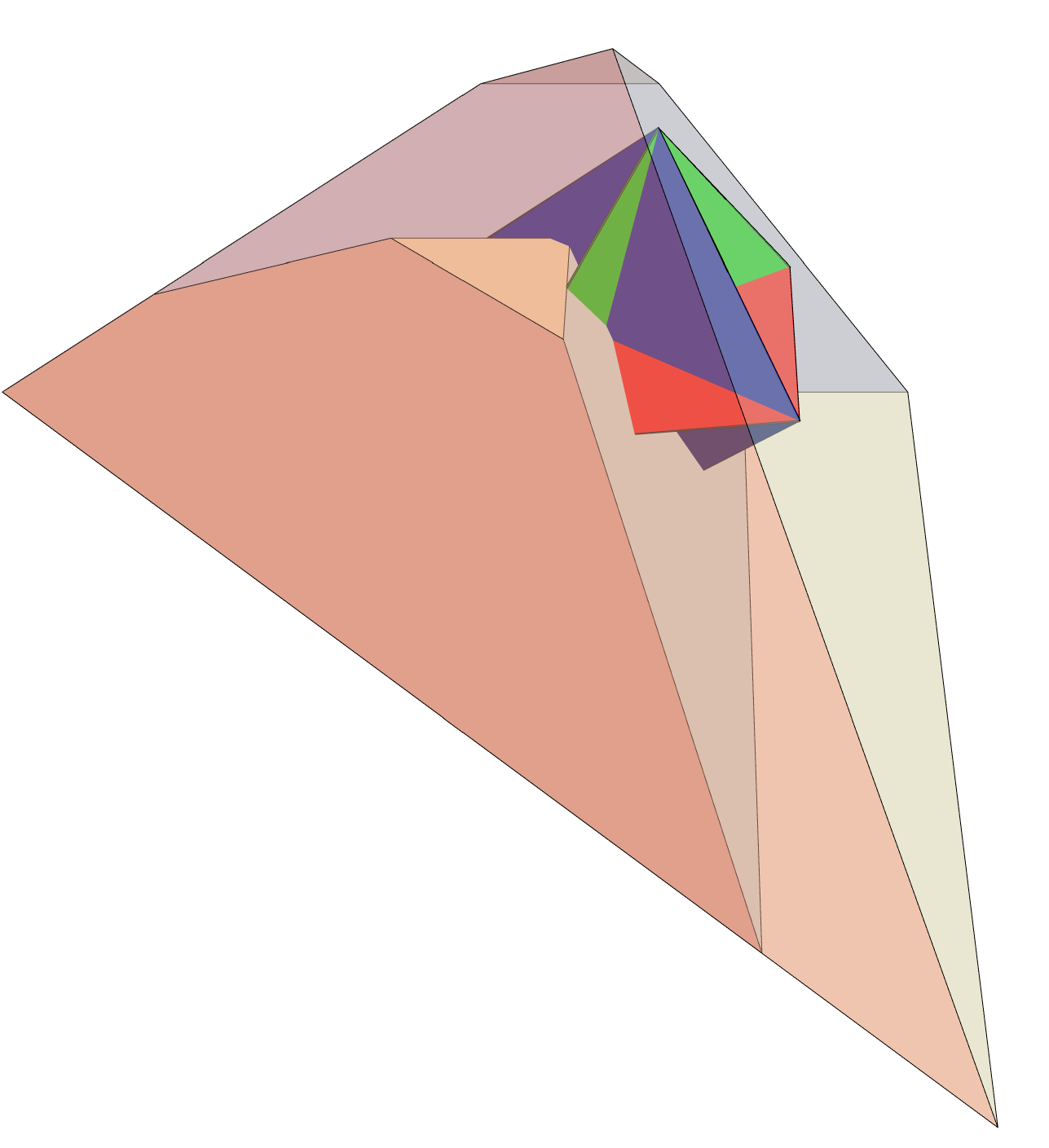}
    \includegraphics[width=0.4\textwidth, trim={6cm 7cm 1.5cm 9cm}, clip]{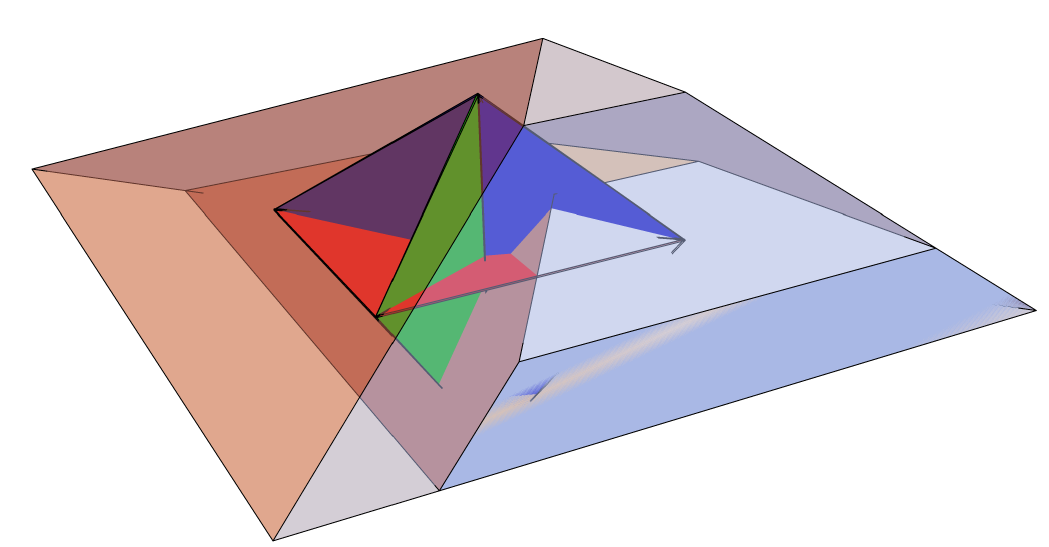}
    \caption{%
    Two perspectives on the configuration of polytopes in \Cref{InterferenceInequalities}.
    The solid, inner polytope is the $\XX$--circuit polytope for a sequence of interaction strengths $(\alpha_1, \ldots, \alpha_n)$, and the translucent outer polytope is the $\XX$--circuit polytope for the sequence $(\alpha_1, \ldots, \alpha_n, \beta)$.
    We have chosen a particular point on the edge of the inner figure and drawn the regions accessible from it by applying \Cref{InterferenceInequalities} for the interaction strength $\beta$ to any pair of coordinates (red: $\XX$ and $\YY$; green: $\YY$ and $\ZZ$; blue: $\XX$ and $\ZZ$).
    These appear as a triple of flat polygons.
    \Cref{MainLocalTheorem} states that when the point on the inner body is permitted to range, the union of the corresponding translates of the red, green, and blue polygons sweeps out the entirety of the outer body.
    }
    \label{InterferencePlanes}
\end{figure}
In \Cref{InterferencePlanes}, we give a visualization of the regions accessible via \Cref{InterferenceInequalities}.
\end{example}

\begin{figure}
    \centering
    \scalebox{0.7}{
    \Qcircuit @C=1.0em @R=0.2em @!R { \\
	 	& \gate{W_0^\dagger \Z_{-r} W_0'{}^\dagger} & \multigate{1}{\CAN(a_1, a_2, a_3)} & \gate{W_0' W_0'' \Z_d} & \multigate{1}{\XX_\beta} & \gate{\Z_{-t} W_0} & \qw \\ 
	 	& \gate{W_1^\dagger \Z_{-s} W_1'{}^\dagger} & \ghost{\CAN(a_1, a_2, a_3)} & \gate{W_1' \Z_e} & \ghost{\XX_\beta} & \gate{\Z_{-u} W_1} & \qw \\
    \\ }}
    \caption{%
    The circuit emitted by a typical single step of \Cref{EffectiveSynthesis} (i.e., by \Cref{MainLocalTheorem}), modeling $\CAN(b_1, b_2, b_3)$ in terms of $\CAN(a_1, a_2, a_3)$, $\XX_\beta$, and local gates.
    The gates $W, W', W''$ are quarter-turns which realize the action of Weyl group elements on $a, b \in \A_{C_2}$.
    }
    \label{SynthesisStepCircuit}
\end{figure}
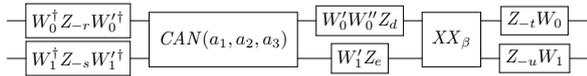

\begin{theorem}[{cf.\ \Cref{SynthesisStepCircuit}}]\label{MainLocalTheorem}
Given a positive canonical triple $(b_1, b_2, b_3)$ satisfying the conditions of \Cref{MainGlobalTheorem} for a sequence of interaction strengths $(\alpha_1, \ldots, \alpha_n, \beta)$, there always exists a positive canonical triple $(a_1, a_2, a_3)$ satisfying the conditions of \Cref{MainGlobalTheorem} for the sequence $(\alpha_1, \ldots, \alpha_n)$ and for which there are Weyl reflections $w$, $w'$ so that the following is solvable:
\begin{align*}
& \hspace{0em} \CAN(a_1, a_2, a_3)^w \cdot (Z_d \otimes Z_e) \cdot \CAN(\beta) \\
& \hspace{12em}\equiv \CAN(b_1, b_2, b_3)^{w'}.
\end{align*}
The outer gates witnessing the local equivalence can be taken to be $Z$-rotations.
\end{theorem}
\begin{proof}
Any canonical gate $\CAN(b_1, b_2, b_3)$ can be written as
\begin{align*}
    \CAN(b_1, b_2, b_3) & = \CAN(0, 0, b_3) \cdot \CAN(b_1, b_2, 0).
    \intertext{%
    Applying \Cref{InterferenceInequalities} to the right factor gives
    }
    \CAN(b_1, b_2, b_3) & = \CAN(0, 0, b_3) \cdot \\
    & \quad \quad \cdot (\Z_{-r} \otimes \Z_{-s}) \cdot \CAN(a_1, a_2, 0) \cdot \\
    & \quad \quad \cdot (\Z_d \otimes \Z_e) \cdot \CAN(\beta) \cdot \\
    & \quad \quad \cdot (\Z_{-t} \otimes \Z_{-u}),
    \intertext{%
    under certain conditions on $a_1$, $a_2$, $b_1$, $b_2$, and $\beta$.
    Since local $Z$--rotations commute with canonical gates of the form $\CAN(0, 0, b_3)$, we may abbreviate this to
    }
    \CAN(b_1, b_2, b_3) & = (\Z_{-r} \otimes \Z_{-s}) \cdot \CAN(a_1, a_2, b_3) \cdot \\
    & \quad \quad \cdot (\Z_d \otimes \Z_e) \cdot \CAN(\beta) \cdot \\
    & \quad \quad \cdot (\Z_{-t} \otimes \Z_{-u}).
\end{align*}
Additionally, our choice to factor out $b_3$ is immaterial: there are Weyl reflections which permute the coordinates within a canonical triple, so by conjugating $\CAN(b_1, b_2, b_3)$ we can place any of the three values in the final slot.
In short, we may appeal to \Cref{InterferenceInequalities}, provided we fix one coordinate and potentially disorder the positive canonical triples.

From here, our proof strategy is similar to that of \Cref{MainGlobalTheorem}.
\Cref{MainGlobalTheorem} itself furnishes us with linear constraints on the spaces of triples $(b_1, b_2, b_3)$ so that a triple satisfies the constraints if and only it can be realized as the positive canonical coordinate of an $\XX$--circuit with interaction strengths $(\alpha_1, \ldots, \alpha_n, \beta)$.
Rather than working with ordered triples $(a_1, a_2, a_3)$, we instead consider unordered triples $(a_h, a_\ell, a_f)$---to be referred to as the ``high'', ``low'', and ``fixed'' coordinates---as in \Cref{DisorderedMainTheorem}.
Then, we interrelate the $a$-- and $b$--coordinates:
\begin{itemize}
    \item
    We select one coordinate $b_f$ from $(b_1, b_2, b_3)$ to serve as the ``fixed'' coordinate (and take the union over such choices), and we set $a_f = b_f$.
    \item
    On $a_h$ and $a_\ell$, we impose the constraint $a_h \ge a_\ell$.
    Similarly, of the remaining coordinates in $(b_1, b_2, b_3)$, we pick $b_h$ to be the larger and $b_\ell$ to be the smaller.
    \item We impose the constraints from \Cref{InterferenceInequalities} on $(a_h, a_\ell, 0)$, $(b_h, b_\ell, 0)$, and $\beta$.
\end{itemize}
Let us call the resulting (nonconvex) polytope $P$.
Points in $P$ capture the following interrelated pieces of data:
\begin{itemize}
    \item
    A canonical coordinate $(b_1, b_2, b_3)$ which admits expression as an $\XX$--circuit with interaction strengths $(\alpha_1, \ldots, \alpha_n, \beta)$.
    \item
    A canonical coordinate $(a_1, a_2, a_3)$ which admits expression as an $\XX$--circuit with interaction strengths $(\alpha_1, \ldots, \alpha_n)$.
    \item
    A choice of value to share among the $a$-- and $b$--coordinates.
    \item
    The condition that, among the unshared coordinates, there exists a circuit of the form in \Cref{XXYYProductsOpaque} relating them.
    (As in the first two bullets, the polytope does \emph{not} record the literal data of such a circuit, only the predicate that one exists.)
\end{itemize}

By projecting away $(a_h, a_\ell, a_f)$ from $P$, we produce the polytope of positive canonical triples $(b_1, b_2, b_3)$ which can be expressed as $\XX$--circuits with the specified interaction strengths, together with the predicate constraint that the last step in the circuit decomposition can be written in the form of the Theorem statement.
This is a subpolytope of that of \Cref{MainGlobalTheorem}, which merely tracks positive canonical triples which can be expressed as $\XX$--circuits with the specified interaction strengths, \emph{without} the constraint on the final local operator.
Appealing again to a computer algebra system, we find that these two polytopes are equal.

See \texttt{regenerate\_xx\_solution\_polytopes} in \texttt{monodromy}~\cite{monodromy} for an executable proof.
\end{proof}

\begin{remark}\label{RelevantConvexSummands}
Naively specified, the polytope $P$ in the proof of \Cref{MainLocalTheorem} has many convex components: the two convex regions of $a$-- and $b$--coordinates each contribute factors of $2$, the choice of which coordinate to fix contributes a factor of $3$, and the choice of which slant and frustrum bounds apply to the disordered $a$--coordinates contribute factors of $2$ and $3$.
However, the projection of $P$ onto the $b$--coordinates, which we used to conclude the theorem, can be shown to have only four regions:
\begin{itemize}
    \item
    The choice of convex region of $b$--coordinates is free, but one then uses the same choice for $a$--coordinates.
    \item
    The fixed coordinate $a_f$ is taken to be either $b_1$ or $b_3$.
    \item
    For the unreflected (resp., reflected) convex region of $b$--coordinates, the slant (resp., strength) inequality is imposed either on $a_f$ or $a_h$ depending on whether $a_f = b_1$ or $a_f = b_3$.
    \item
    The frustrum bound is always imposed on $a_\ell$.
\end{itemize}
The inequalities describing these regions are given in \Cref{ProjectedInequalityTables}.
\end{remark}

\begin{remark}\label{TechnophobeRemark}
It is possible for the technophobic reader to rearrange the proofs of \Cref{MainGlobalTheorem} and \Cref{MainLocalTheorem} so as to avoid computer algebra systems.
First, break \Cref{MainGlobalTheorem} into a forward implication, that the positive canonical triple associated to an $\XX$--circuit satisfies the indicated inequality set, and the reverse implication.
The forward implication can be checked by hand, using a judiciously chosen subset of inequalities from the monodromy polytope; the reverse implication is much harder from this point of view, so we set it aside for a moment.

Now we turn to \Cref{MainLocalTheorem}.
Its proof also relies on a computer algebra system, but we may severely limit the amount of work by inspecting only the convex summands described in \Cref{RelevantConvexSummands}, which is then small enough to accomplish manually.
With only the forward implication of \Cref{MainGlobalTheorem} established, the proof of \Cref{MainLocalTheorem} instead shows that if the $b$--coordinate belongs to the polytope named by \Cref{MainGlobalTheorem} for $(\alpha_1, \ldots, \alpha_n, \beta)$, then there exists an $a$--coordinate in the polytope named by \Cref{MainGlobalTheorem} for $(\alpha_1, \ldots, \alpha_n)$ which is related to the $b$--coordinate by a particular single-step $\XX$--circuit.
Following the induction described in \Cref{EffectiveSynthesis} then yields the missing reverse implication of \Cref{MainGlobalTheorem}, which in turn yields the full strength of \Cref{MainLocalTheorem}.
\end{remark}

\section{Optimal synthesis}\label{OptimalSynthesisSection}

We now put the pieces together to form an optimal synthesis routine.
The actual synthesis process is now straightforward, given in \Cref{EffectiveSynthesis}, but it is trickier to pin down exactly what is meant by ``optimal''.
For instance, the notion of optimality considered by Zhang et al.~\cite{ZVSWControlled} is to minimize two-qubit operation count---but in a larger gateset, where different gates may have uneven performance impact, optimizing count alone may not optimize performance.
Relatedly, if performance is the true goal and the performance penalty incurred for using gates is high, it may be preferable to synthesize a circuit modeling some canonical triple $a' \ne a$ which requires fewer gates, trading the performance hit due to the mismatch for performance gain of dropping some of the gates.

Let us begin with the synthesis procedure itself:

\begin{procedure}[{cf.\ \Cref{SynthesisStepCircuit}}]\label{EffectiveSynthesis}
The existence claim of \Cref{MainLocalTheorem} can be promoted into an algorithmically effective synthesis routine.
Given a sequence of interaction strengths $(\alpha_1, \ldots, \alpha_n, \beta)$ and a positive canonical triple $(b_1, b_2, b_3)$ which belongs to the associated circuit polytope, the polytope $P$ from the proof of \Cref{MainLocalTheorem} can then be specialized so that only $a_h$ and $a_\ell$ are free variables.
(We report these inequality sets in \Cref{LiftedInequalityTables}.)
The content of \Cref{MainLocalTheorem} is that this specialization is always nonempty, so we may find a point $(a_h, a_\ell, a_f)$ in it (e.g., by calculating line-line intersections until we produce a vertex).
This pair of points can then be fed to \Cref{InterferenceExpressions}, which produces the angle values for the $\Z$--rotations.
This proceeds recursively until the sequence of interaction strengths is exhausted.
\end{procedure}

\begin{example}
\begin{figure}
    \centering
    \includegraphics[width=0.4\textwidth, trim={2.8cm 0 3.7cm 0}, clip]{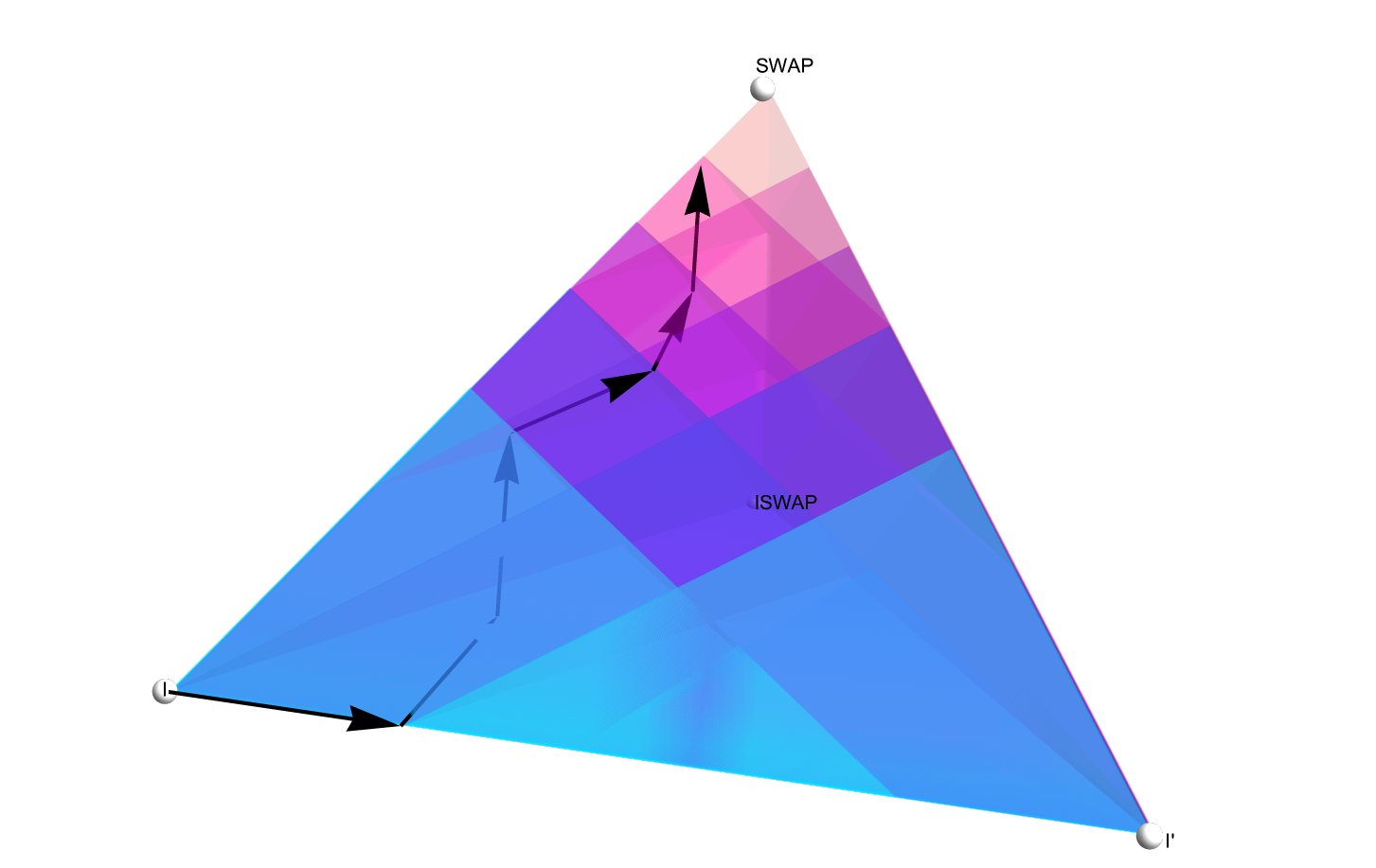}
    \includegraphics[width=0.4\textwidth, trim={5cm 2cm 1cm 3.5cm}, clip]{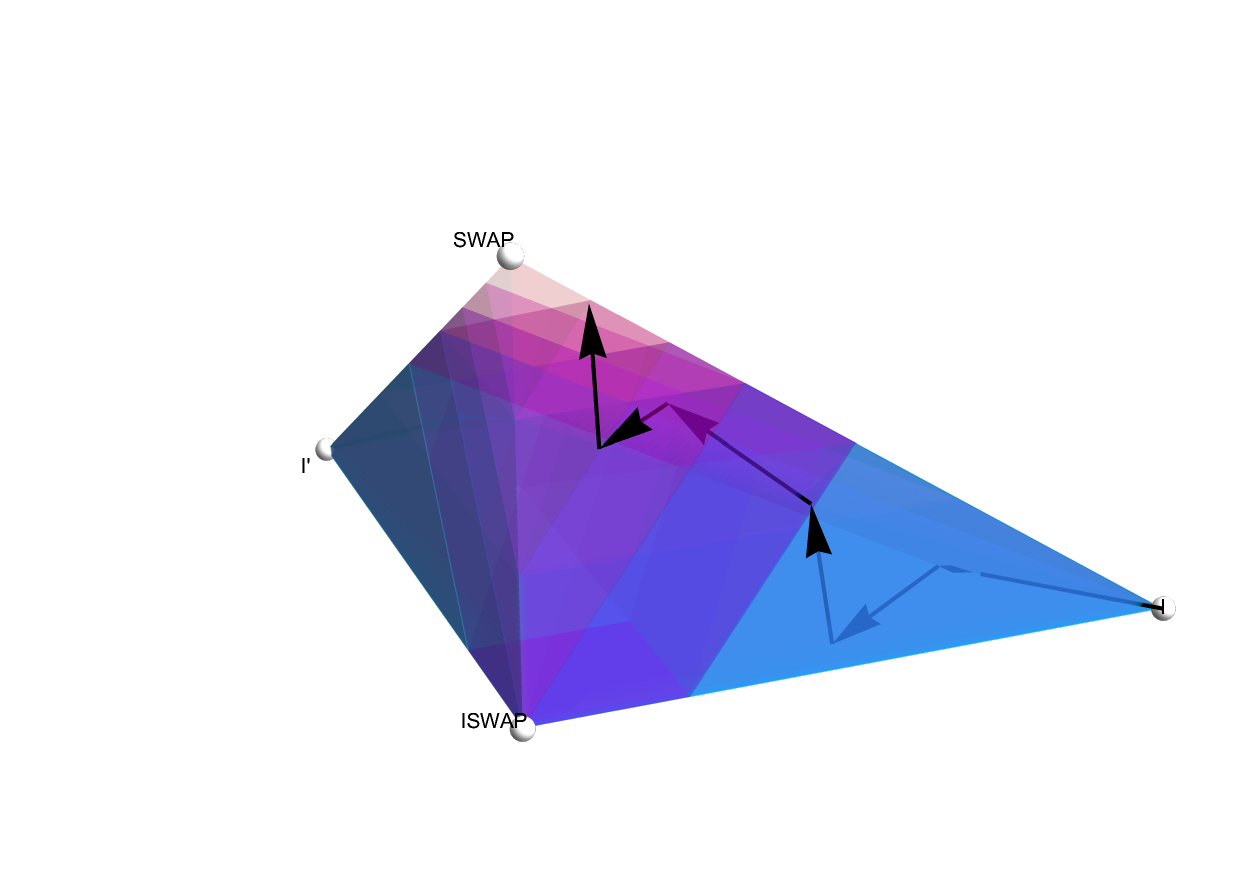}
    \caption{%
    Path of intermediate points in $\A_{C_2}$ produced by \Cref{EffectiveSynthesis} synthesizing $\CAN(0.707, 0.687, 0.687)$ into an $\XX$--circuit with strength sequence $(\frac{\pi}{8}, \frac{\pi}{8}, \frac{\pi}{12}, \frac{\pi}{12}, \frac{\pi}{12}, \frac{\pi}{12})$.
    The various colored regions are the circuit polyhedra for truncations of this sequence of interaction strengths.
    }
    \label{PathFigure}
\end{figure}
In \Cref{PathFigure}, we include a visualization of the intermediate steps produced when using \Cref{EffectiveSynthesis} to synthesize an $\XX$--circuit for a certain canonical point against a particular sequence of interaction strengths.
\end{example}

To progress, we need a quantitative definition of optimality.

\begin{definition}\label{AbstractCostFnDefinition}
Given a target unitary $U$, a gate set $\S$, and a cost function $\C_{\S}$ which consumes $U$ and an $\S$--circuit $C$, the \textit{approximate synthesis task} is to produce an $\S$--circuit $C$ maximizing $\C_{\S}(U, C)$.
Cost functions can enjoy a variety of pleasant properties:
\begin{description}
    \item[\textit{Separable}:]
    For a circuit template $C(\mathbf \theta)$, $\C_{\S}(U, C)$ can be written as a sum $\C'(U, C(\mathbf \theta)) + \C''_{\S}(C, \mathbf \theta)$ where $\C'$ depends only on the unitary which the circuit models and $\C''$ depends only on the circuit and parameters---but not its relationship to $U$.
    \item[\textit{Locally invariant}:]
    $\C''_{\S}(C, \theta) = \C''_{\S}(C)$ is invariant under choice of parameters $\theta$ for local gates in the circuit $C$.
    \item[\textit{Monotonic}:]
    Suppose that $\C$ is a separable cost function.
    If $C$ is contained as a subcircuit in $D$, then $\C''_{\S}(C, \theta) \le \C''_{\S}(D, (\theta, \phi))$.
    \item[\textit{Non-approximating}:]
    The separable cost function $\C$ has $\C'$ given by \[\C'(U, V) = \begin{cases}0 & \text{when $U = V$}, \\ \infty & \text{otherwise} .\end{cases}\]
\end{description}
\end{definition}

These features are chosen both because they feed into an efficient algorithm for optimal synthesis and because they are satisfied in the following guiding example:

\begin{example}\label{InfidelityCostFn}
The \textit{average infidelity} of two gates $U$ and $V$ is
\begin{align*}
    I(U, V) & = 1 - \int_{\psi \in \mathbb P(\mathbb C^4)} \<\psi| U^\dagger V |\psi\>^2 \\
    & = \frac{16 - |\tr U^\dagger V|^2}{4 \cdot 5} \in [0, 4/5].
\end{align*}
For $\S$ a finite collection of $\XX$--interactions with costs $c\co \S \to \R$, we define a separable, locally invariant, monotonic cost function by \[\C_{\S}(U, C) = I(U, C) + \sum_{s \in C} c(s).\]
\end{example}

Average gate infidelity satisfies a few pleasant generic properties, but it is also tightly connected to the theory of $KAK$ decompositions.
We record these properties below.

\begin{remark}\label{InfidelityProperties}
Average infidelity detects gate equivalence, in the sense that $I(U, V) = 0$ if and only if $U = V$.
It is also symmetric: $I(U, V) = I(V, U)$.
However, it fails to satisfy the triangle inequality, even when $U$ and $V$ belong to the canonical family, hence does \emph{not} give a metric.
It satisfies compositionality only to first order:
\begin{align*}
    & \hspace{-1em} 16 - 20 \cdot I(UU', (U + \eps E)(V + \zeta F)) \\
    & = |\tr V^\dagger U^\dagger (U + \eps E)(V + \zeta F)|^2 \\
    & = |1 + \eps \tr U^\dagger E + \zeta \tr V^\dagger F + \eps \zeta \tr V^\dagger U^\dagger E F|^2.
\end{align*}
\end{remark}

\begin{lemma}[{\cite{CBSNG}}]\pushQED{\qed}
Let $U = \CAN(a_1, a_2, a_3)$ and $V = \CAN(b_1, b_2, b_3)$ be two canonical gates with parameter differences $\delta_j = (a_j - b_j)$.
Their average gate infidelity is given by \[20 \cdot I(U, V) = 16 - 16\left(\prod_j \cos^2 \frac{\delta_j}{2} + \prod_j \sin^2 \frac{\delta_j}{2} \right). \qedhere \popQED\]
\end{lemma}

\begin{lemma}[{\cite{WVMCWRGK}}]\label{ApproximationOuterGates}
Suppose that $C_1$, $C_2$ are fixed canonical gates and that $L_1$, $L_1'$ are fixed local gates.  Letting $L_2$ and $L_2'$ range over all local gates, the value $I(L_1 C_1 L_1', L_2 C_2 L_2')$ is minimized when taking $L_2 = L_1$ and $L_2' = L_1'$. \qed
\end{lemma}

We now describe an optimal synthesis procedure for a nice cost function:

\begin{procedure}\label{SynthesisWrapper}
Let $\C$ be a separable, locally invariant, and monotonic cost function.
Let $\S$ be a finite gate set of $\XX$--type interactions, and consider the set of circuit templates given by interleaving unconstrained local gates into the words formed from $\S$.
Traverse these available circuit templates (i.e., the words in $\S$) by ascending order of $\C''_{\S}$.%
\footnote{%
Using a priority queue, one can perform this traversal without enumerating all possible words beforehand.
}
For each such circuit template $C$, use \Cref{MainGlobalTheorem} to calculate the circuit polytope $\Pi(C)$.
Calculate the point $p \in \Pi(C)$ which optimizes $\C'(U, \CAN(p))$.
If the total cost $\C_{\S}$ is the best seen so far, retain $C$ and $p$.
Continue to traverse circuit templates until $\Pi(U) \in \Pi(C)$, at which point $\C'$ vanishes and the ordering of circuit templates guarantees that all future circuit templates will yield a worse cost.%
\footnote{%
\Cref{MainGlobalTheorem} guarantees that this termination condition will eventually be met provided $\S$ contains any interaction $\XX_\beta$ with $\beta \in (0, \frac{\pi}{2})$.
}
Finally, apply \Cref{EffectiveSynthesis} to synthesize a $C$--circuit for $\CAN(p)$, then apply \Cref{ApproximationOuterGates} to produce $U$ itself.
\end{procedure}

\begin{remark}
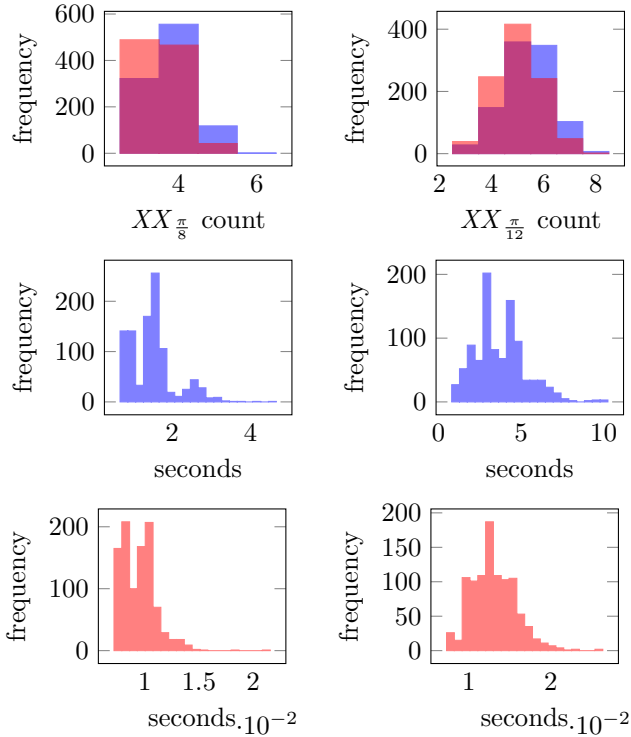
\begin{figure}
    \centering
    \begin{tabular}{cc}
    {\begin{tikzpicture}
    \begin{axis}[%
        xtick style={draw=none},
        scale=0.36,
        xlabel={$\XX_{\frac{\pi}{8}}$ count},
        ylabel={frequency},
        xlabel near ticks,
        ylabel near ticks,
    ]
    \addplot +[hist={bins=4, data min=2.5, data max=6.5}, fill, opacity=0.5, mark=none] table[y=nuop_depth]{data/nuop_stats_pi4.dat};
    \addplot +[hist={bins=3, data min=2.5, data max=5.5}, fill, opacity=0.5, mark=none] table[y=monodromy_depth]{data/nuop_stats_pi4.dat};
    \end{axis}
    \end{tikzpicture}}
    &
    {\begin{tikzpicture}
    \begin{axis}[%
        xtick style={draw=none},
        scale=0.36,
        xlabel={$\XX_{\frac{\pi}{12}}$ count},
        ylabel={frequency},
        xlabel near ticks,
        ylabel near ticks,
    ]
    \addplot +[hist={bins=6, data min=2.5, data max=8.5}, fill, opacity=0.5, mark=none] table[y=nuop_depth]{data/nuop_stats_pi6.dat};
    \addplot +[hist={bins=6, data min=2.5, data max=8.5}, fill, opacity=0.5, mark=none] table[y=monodromy_depth]{data/nuop_stats_pi6.dat};
    \end{axis}
    \end{tikzpicture}}
    \\
    {\begin{tikzpicture}
    \begin{axis}[%
        xtick style={draw=none},
        scale=0.36,
        xlabel={seconds},
        ylabel={frequency},
        xlabel near ticks,
        ylabel near ticks,
    ]
    \addplot +[hist={bins=20}, blue, fill, opacity=0.5, mark=none] table[y=nuop_time]{data/nuop_stats_pi4.dat};
    \end{axis}
    \end{tikzpicture}}
    &
    {\begin{tikzpicture}
    \begin{axis}[%
        xtick style={draw=none},
        scale=0.36,
        xlabel={seconds},
        ylabel={frequency},
        xlabel near ticks,
        ylabel near ticks,
    ]
    \addplot +[hist={bins=20}, blue, fill, opacity=0.5, mark=none] table[y=nuop_time]{data/nuop_stats_pi6.dat};
    \end{axis}
    \end{tikzpicture}}
    \\
    {\begin{tikzpicture}
    \begin{axis}[%
        xtick style={draw=none},
        scale=0.36,
        xlabel={seconds},
        ylabel={frequency},
        xlabel near ticks,
        ylabel near ticks,
    ]
    \addplot +[hist={bins=20}, red, fill, opacity=0.5, mark=none] table[y=monodromy_time]{data/nuop_stats_pi4.dat};
    \end{axis}
    \end{tikzpicture}}
    &
    {\begin{tikzpicture}
    \begin{axis}[%
        xtick style={draw=none},
        scale=0.36,
        xlabel={seconds},
        ylabel={frequency},
        xlabel near ticks,
        ylabel near ticks,
    ]
    \addplot +[hist={bins=20}, red, fill, opacity=0.5, mark=none] table[y=monodromy_time]{data/nuop_stats_pi6.dat};
    \end{axis}
    \end{tikzpicture}}
    \end{tabular}
    \caption{%
    Comparison of the output and wall-time characteristics of our algorithm (in red) and that of a \texttt{numpy} numerical search (in blue), when targeting the two-qubit gate sets $\S = \{\XX_{\frac{\pi}{8}}\}$ and $\S = \{\XX_{\frac{\pi}{12}}\}$.
    We plot the histograms of wall-times separately, as numerical search is $> 200\times$ slower. Numerical search can fall in local wells and produce sub-optimal circuits. 
    }
    \label{SynthesisWallTimeFig}
\end{figure}
In \Cref{SynthesisWallTimeFig}, we study the execution characteristics of \Cref{SynthesisWrapper} compared to those of blind numerical search.
The implementation of our method is available in Qiskit's \texttt{quantum\_info} subpackage as the class \texttt{XXDecomposer}~\cite{Qiskit}.
Given a Haar-randomly chosen two-qubit unitary operator $U$, our numerical search procedure is to let \texttt{numpy}'s generic optimizer explore the space of circuits of a particular depth, with the objective of minimizing the infidelity with $U$.
If the optimizer cannot find a circuit with infidelity below some threshold, we retry with a circuit of the next larger depth.
Altogether, this is similar to what is implemented in \texttt{NuOp}~\cite{LMMB}, among other compilation suites.
The histograms reported in \Cref{SynthesisWallTimeFig} are the result of sampling over many such $U$, targeting either the gate set $\S = \{\XX_{\frac{\pi}{8}}\}$ or $\S = \{\XX_{\frac{\pi}{12}}\}$.%
\footnote{%
Neither our implementation of \Cref{SynthesisWrapper} nor our invocation of \texttt{numpy} is particularly clever.
We expect that both distributions can be shifted left with further optimization of the implementations, but that the multiplicative difference will be at least as large between ``optimal'' implementations of each synthesis method.
}
\end{remark}

It remains to describe how to find the point $p \in \Pi(C)$ which optimizes $\C'(U, \CAN(p))$.
For a non-approximating cost function, this can be probed directly: if $\Pi(U) \in \Pi(C)$, then we take $p = \Pi(U)$, and otherwise we reject $\Pi(C)$ entirely.
For the approximating cost function defined in \Cref{InfidelityCostFn}, we use the following more elaborate result:

\begin{theorem}\label{ApproximationTheorem}
Let $P$ be an $\XX$--circuit polytope, and let $F \subseteq P$ be an open facet within it.
For $p \in \A_{C_2}$ a fixed positive canonical triple, if $b = q$ is a critical point of the infidelity distance $I|_{a = p, b \in F}$ as constrained to $F$, then $q$ is also a critical point of the Euclidean distance to $p$ as constrained to $F$.%
\end{theorem}
\begin{proof}
\Cref{MainGlobalTheorem} gives an explicit enumeration of the available open facets of $P$, and we approach the optimization problem over each facet separately.
We defer this to \Cref{ProofOfApproximationTheorem}.
\end{proof}


\noindent
This result means that we can repurpose the standard procedure used to calculate the nearest point in Euclidean distance to instead find the best approximating canonical triple.
Namely, to calculate the nearest point in Euclidean distance, project the point onto the affine subspaces spanned by each facet of the polytope (e.g., by solving a least-squares problem), retain those projections which belong to the polytope, and from that finite set select the point of minimum (infidelity) distance.

\begin{remark}
This is extremely unusual behavior for these two optimization problems and relies on the specific form of the polytopes appearing in \Cref{MainGlobalTheorem}.
For contrast, consider the line passing through the origin with slope $(\frac{\pi}{4}, \frac{\pi}{50}, \frac{\pi}{50})$ and the off-body point $(\frac{83 \pi}{400}, \frac{83 \pi}{400}, \frac{83 \pi}{400})$.
The fidelity-nearest point appears after traveling for one unit of time, but the Euclidean-nearest point appears after traveling for $\approx 95\%$ of a unit of time.
\end{remark}

\begin{remark}
Numerical experiment indicates that the nearest point under infidelity distance exactly agrees with the nearest point under Euclidean distance---i.e., that the same critical point achieves the minimum value in both of these searches.
However, this conjecture yields no algorithmic speedup when producing these minimizers, so we are not motivated to pursue it here.
\end{remark}

\section{Gateset optimization and numerical experiment}\label{GatesetOptSection}

In this section, we bring the theory of \Cref{OptimalSynthesisSection} to bear on deciding which native gates are worth bestowing on a device.
Even if a device is physically capable of enacting some quantum operation, there is calibration overhead to making that operation available as a reliable user-facing gate.
At the same time, the more high-fidelity native interactions are available, the more clever and adaptable a synthesis method (including ours) can be.
Accordingly, we would like to find a small set of $\XX$ operations that optimizes certain objective functions which measure synthesis performance.
The primary objective with which we will concern ourselves is expected cost:

\begin{definition}
For a two-qubit unitary $U \in \PU(4)$ and a native gate set $\S$, let $\C_{\S}(U) := \min_C \C_{\S}(U, C)$ be a cost function as in \Cref{AbstractCostFnDefinition} (e.g., \Cref{InfidelityCostFn} or its non-approximating variant).
The \textit{expected cost} is defined as \[\<\C_{\S}\> = \int_{U \in \PU(4)} \C_{\S}(U) \, d\mu^{\Haar}.\]
\end{definition}

For $\XX$--based gate sets $\S$ and for favorable cost functions, we now show how to compute this value \emph{exactly}.
Starting with the definition
\[\<\C_{\S}\> = \int_{U \in PU(4)} \min_C \C_{\S}(U, C) \, d\mu^{\Haar},\]
we use separability and non-approximation to reduce to the case where $U$ admits an exact model by $C$:
\[\<\C_{\S}\> = \int_{U \in PU(4)} \min_{U = C(\theta)} \C''_{\S}(C, \theta) \, d\mu^{\Haar}.\]
By assuming $\S$ finite and $\C''_{\S}$ locally invariant, we learn that the integrand $\min_{U \in C} \C''_{\S}(C)$ takes on finitely many values, supported by finitely many choices of $C$.
By sorting the $C$ compatibly with $\C''_{\S}(C)$, we may further reduce to
\[\<\C_{\S}\> = \sum_C \int_{\substack{U \in \operatorname{Image}(C) \\ U \not \in \operatorname{Image}(C' \mid C' < C)}} \C''_{\S}(C) \, d\mu^{\Haar}.\]
Since $\C''_{\S}(C)$ is constant on each region, each summand is given by the reweighted Haar volume of the corresponding region.
Since constant functions pull back from constant functions, we can also push these integrals forward along $\Pi$ and compute them in $\A_{C_2}$:
\[\<\C_{\S}\> = \sum_C \C''_{\S}(C) \left(\Pi_* \mu^{\Haar}\left( \Pi(C) \setminus \bigcup_{C' < C} \Pi(C') \right) \right).\]

Altogether, this reduces the problem to calculating the Haar volume of the polytopes which appear in \Cref{MainGlobalTheorem}.
A formula for $\Pi_* \mu^{\Haar}$ which enables this was previously reported by Watts et al.:

\begin{lemma}[{\cite{WCV}, \cite{MKMZ}, \cite{Mehta}}]\pushQED{\qed}
The pushforward of the Haar measure is given by%
\footnote{%
The extra factor of $2$ appearing in this formula comes from a different scaling of our coordinate systems.
}%
\textsuperscript{,}%
\footnote{%
This density function has a unique local maximum at $(\frac{\pi}{4}, \frac{\pi}{8}, 0)$.
}
\[\Pi_* d\mu^{\Haar} = \frac{384}{\pi} \prod_{1 \le j < k \le 3} \sin(2 c_j + 2 c_k) \sin(2 c_j - 2 c_k). \qedhere \popQED\]
\end{lemma}

\begin{figure}[t]
    \centering
    \begin{tikzpicture}[scale=0.90]
    \begin{axis}[
    xlabel={CX fractional strength},
    ylabel={infidelity},
    legend entries = {average cost},
    legend pos=south east,
    xlabel near ticks,
    ylabel near ticks,
    scale=0.9,
    ]
    \addplot +[mark=none] table[x expr={\thisrow{strength1}*3.1415/4},y=average_cost]{data/gateset_landscape_1d.dat};
    \end{axis}
    \end{tikzpicture}
    \caption{%
    The expected infidelity of a Haar-random operator decomposed exactly into $\S_x = \{\CX, \XX_x\}$.
    Uses an additive affine error model with offset $b \approx 1.909 \times 10^{-3}$ and slope $\frac{\pi}{4} \cdot m \approx 5.76 \times 10^{-3}$.
    }
    \label{CostGraph1D}
\end{figure}
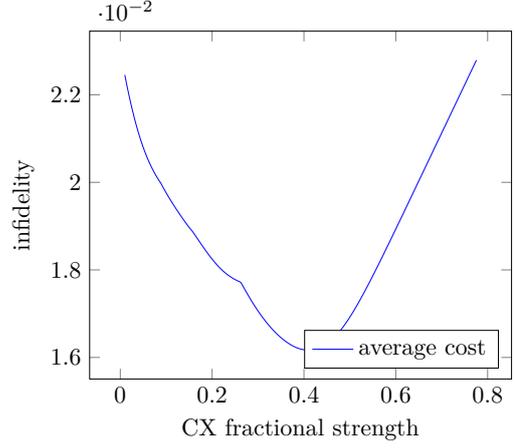

\noindent
Such trigonometric integrals over tetrahedra can be performed exactly.
Altogether, this gives us quantitative means by which to study the effect of tuning the inputs to a parametric gate set, e.g., $\S(x) = \{\XX_{\frac{\pi}{4}}, \XX_x\}$.
A parametric choice of gate set requires a parametric cost function, and our parametric cost function of interest is as follows:

\begin{definition}\label{AffineErrorModelDefn}
In our setting, we find it experimentally justified to assume an \textit{affine error model}: we take $\XX_x$ to have fidelity cost $mx + b$ for some experimentally determined values of $m$ and $b$.%
\footnote{%
In one experiment, we measured $\frac{\pi}{4} \cdot m \approx 5.76 \times 10^{-3}$ and $b \approx 1.909 \times 10^{-3}$.
This reported offset $b$ incorporates the average infidelity cost of local post-rotations, so as to better model the total circuit execution cost while maintaining local invariance.
}
From this, we build a separable, locally invariant, additive cost component by \[\C''_{\S}(C) = \sum_{\XX_x \in C} (mx + b).\]
\end{definition}

\begin{remark}
The reader who would like to account, in the above framework, for the \emph{worst case} cost of the interleaved single-qubit operations can absorb that extra amount into the $b$ parameter.
\end{remark}

\begin{example}\label{Cost1DExample}
\begin{figure}[t]
    \centering
    \includegraphics[width=0.4\textwidth, trim={2cm 0cm 8cm 3cm}, clip]{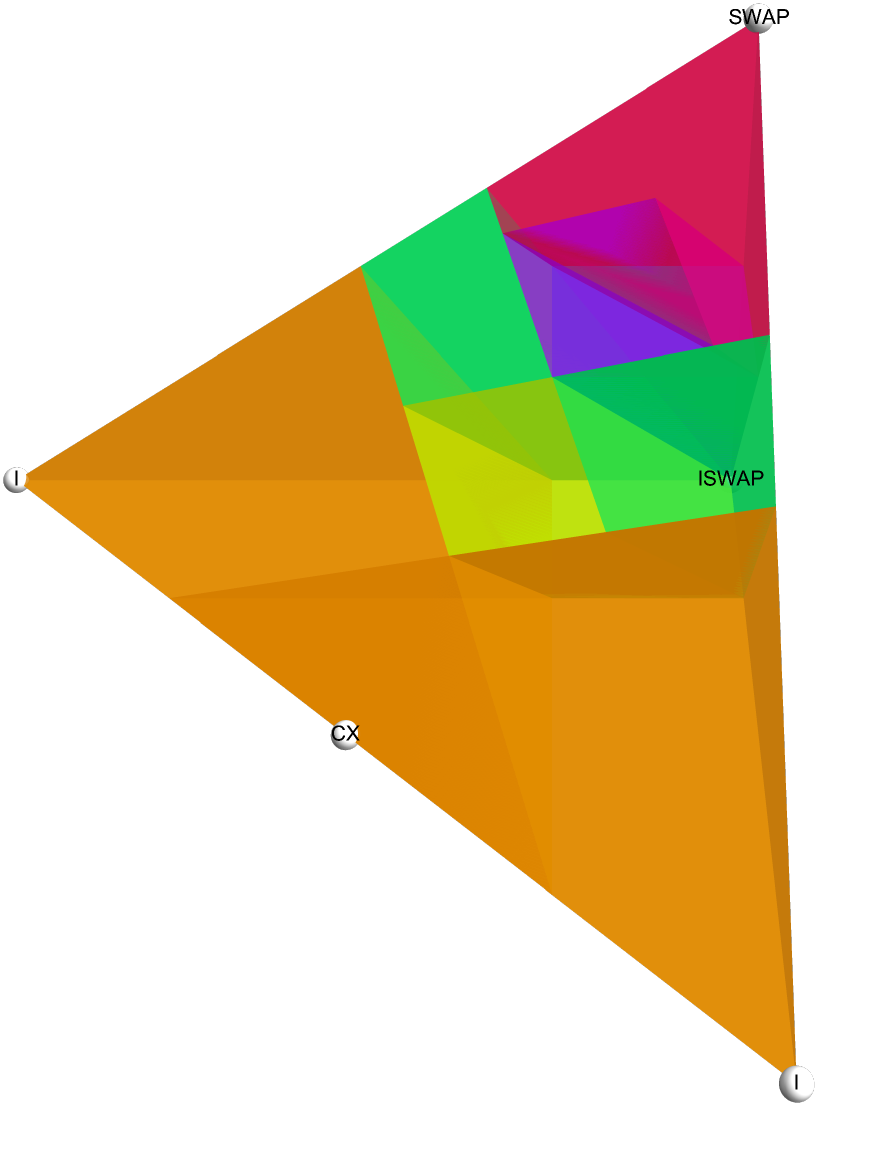}
    \includegraphics[width=0.4\textwidth, trim={5.5cm 4.5cm 4cm 9cm}, clip]{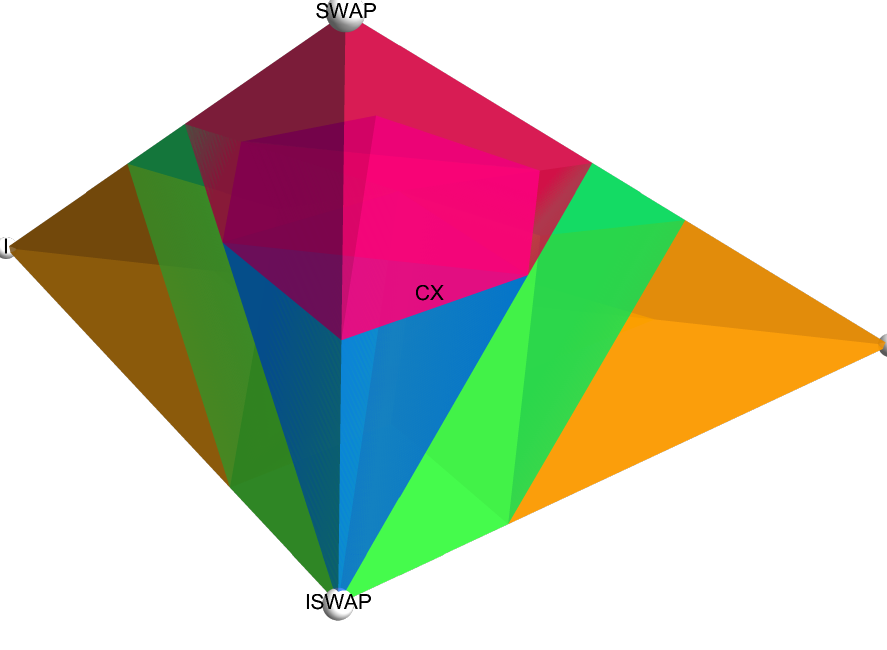}
    \caption{%
    An optimal set of $\S$--circuit polytopes covering $\A_{C_2}$ for $\S = \{\XX_{\frac{\pi}{4}}, \XX_{\frac{\pi}{8}}\}$.
    There are six regions depicted: $(\frac{\pi}{8}, \frac{\pi}{8}, \frac{\pi}{8})$ in orange, $(\frac{\pi}{8}, \frac{\pi}{8}, \frac{\pi}{4})$ in yellow, $(\frac{\pi}{8}, \frac{\pi}{8}, \frac{\pi}{8}, \frac{\pi}{8})$ in green, $(\frac{\pi}{8}, \frac{\pi}{4}, \frac{\pi}{4})$ in blue, $(\frac{\pi}{8}, \frac{\pi}{8}, \frac{\pi}{8}, \frac{\pi}{4})$ in purple, and $(\frac{\pi}{4}, \frac{\pi}{4}, \frac{\pi}{4})$ in red.
    There are also six regions which have circuit depth at most two, hence they do not contribute volume and we suppress them from the picture.
    }
    \label{CXCX2CoverageSets}
\end{figure}
Consider the gate set $\S = \{\XX_{\frac{\pi}{4}}, \XX_x\}$ with cost given by the additive affine error model with parameters $b \approx 1.909 \times 10^{-3}$ and $\frac{\pi}{4} \cdot m \approx 5.76 \times 10^{-3}$.
In \Cref{CostGraph1D}, we display how the expected infidelity of synthesizing an $\S$--circuit for a Haar-randomly chosen unitary varies with the gate set parameter $x$.
The ends of this curve degenerate to the case of the smaller gate set $\{\XX_{\frac{\pi}{4}}\}$.
The precise location of the optimum in the middle depends on the ratio $m / b$; for experimentally realistic error models like the one depicted here, it is located near $\frac{\pi}{8}$, achieving an expected infidelity of $1.62 \times 10^{-2}$.
We observe also that the basin for this minimum is fairly wide, so that $\frac{\pi}{8}$ is a good choice for inclusion in a native gate set even if the error model varies somewhat over time or across a device.%
\footnote{%
Low-denominator rational multiples of $\frac{\pi}{4}$ are also easier to use in a randomized benchmarking scheme.
}
Finally, in \Cref{CXCX2CoverageSets} we depict the optimal synthesis regions within the Weyl alcove for the gate set $\{\XX_{\frac{\pi}{4}}, \XX_{\frac{\pi}{8}}\}$.
\end{example}

\begin{example}\label{Cost2DExample}
\begin{figure}[t]
    \centering
    \begin{tikzpicture}[scale=0.90]
    \begin{axis}[
    xlabel={CX fractional strength},
    ylabel={CX fractional strength},
    zlabel={infidelity},
    legend entries = {average cost},
    legend pos=north west,
    colorbar,
    ]
    \addplot+ [only marks, scatter, point meta=explicit] table[x expr={\thisrow{strength1}*3.1415/4},y expr={\thisrow{strength2}*3.1415/4},meta expr={\thisrow{average_cost}}]{data/gateset_landscape_2d.dat};
    \end{axis}
    \end{tikzpicture}
    \caption{%
    The expected infidelity of a Haar-random operator decomposed exactly into $\S_{x, y} = \{\CX, \XX_x, \XX_y\}$.
    Uses an additive affine error model with offset $b \approx 1.909 \times 10^{-3}$ and slope $\frac{\pi}{4} \cdot m \approx 5.76 \times 10^{-3}$.
    }
    \label{CostGraph2D}
\end{figure}
\begin{figure}[t]
    \centering
    \includegraphics[width=0.4\textwidth, trim={2cm 0cm 8cm 3cm}, clip]{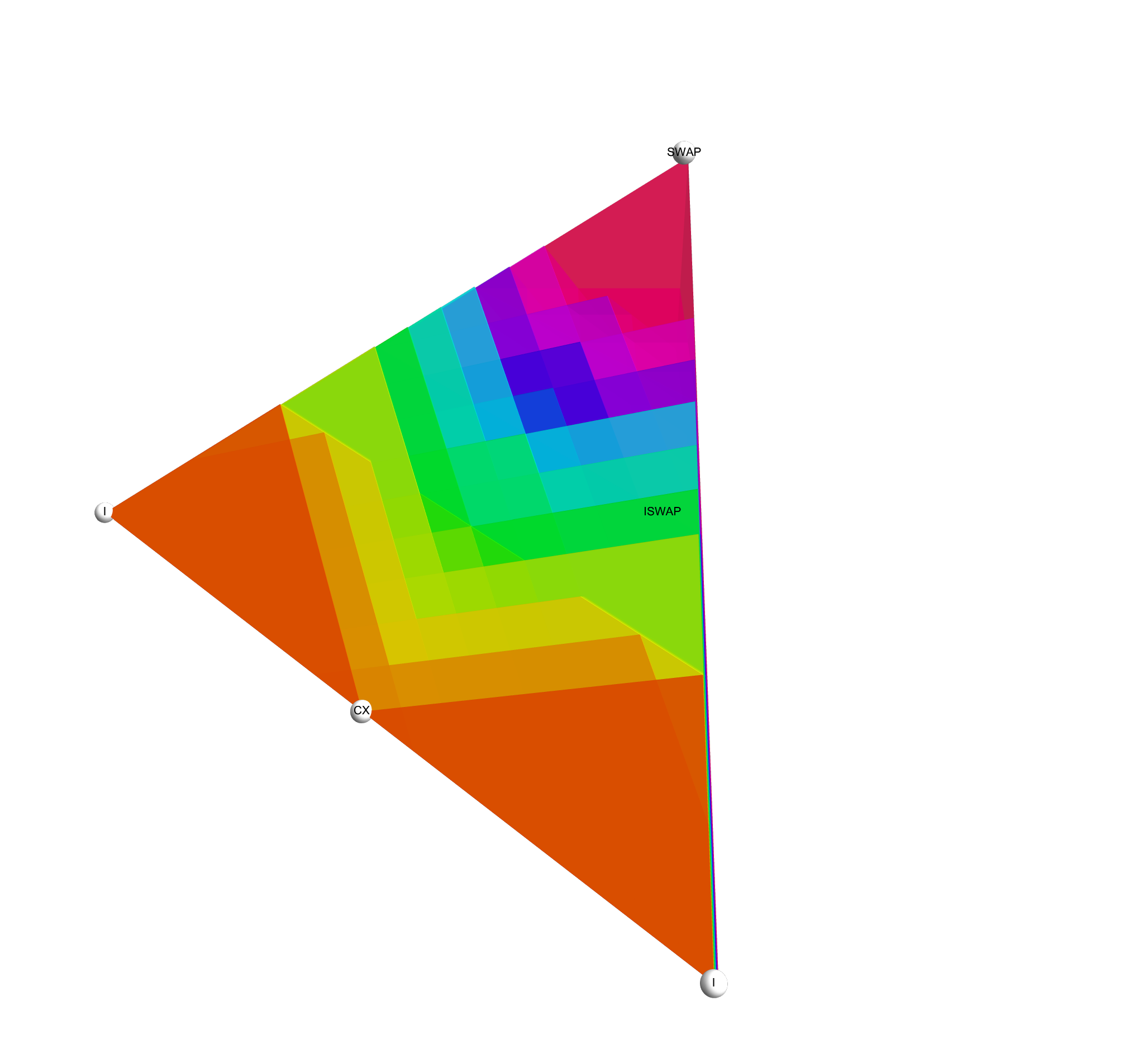}
    \includegraphics[width=0.4\textwidth, trim={8.5cm 6.5cm 5cm 11.5cm}, clip]{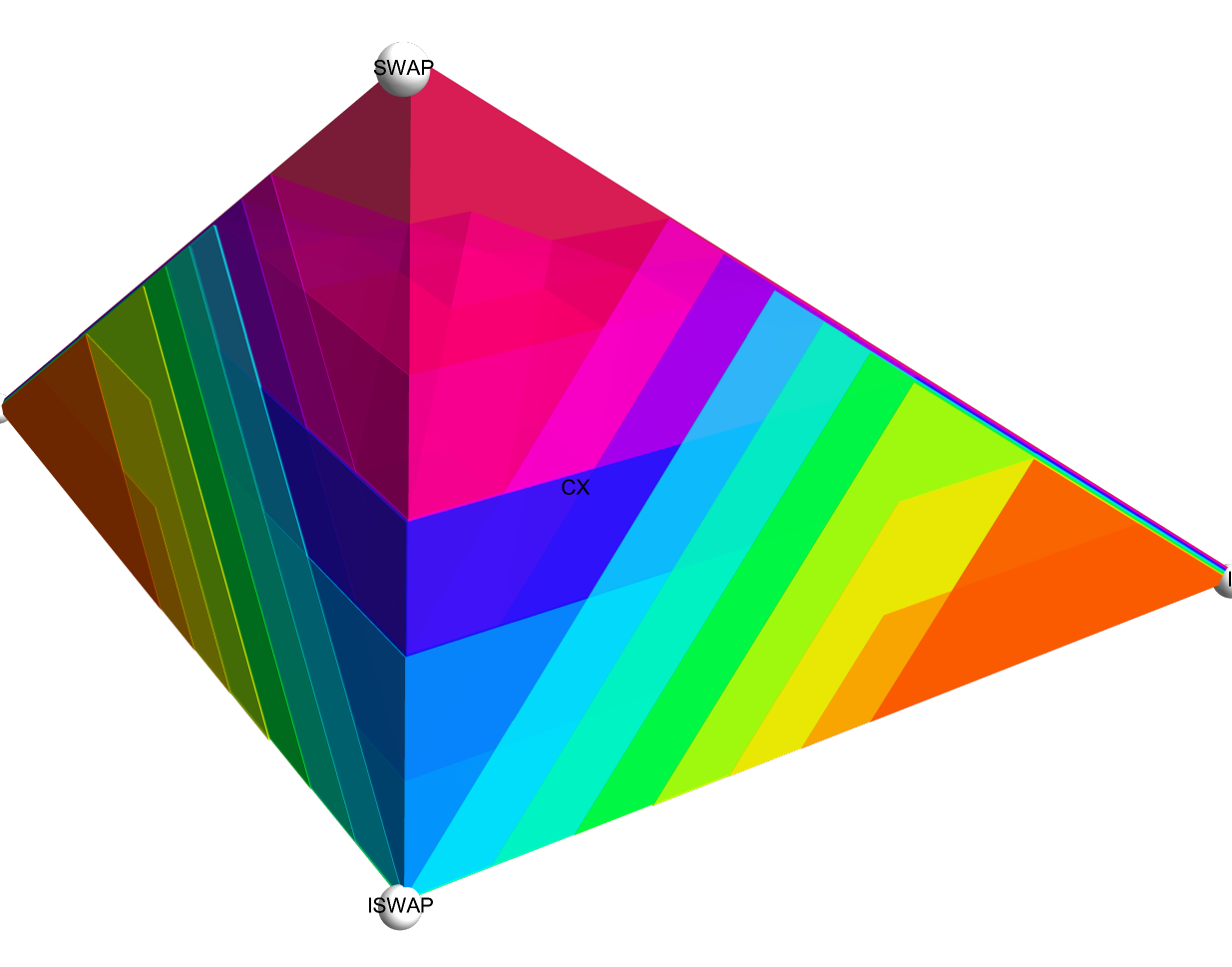}
    \caption{%
    An optimal set of $\S$--circuit polytopes covering $\A_{C_2}$ for $\S = \{\XX_{\frac{\pi}{4}}, \XX_{\frac{\pi}{8}}, \XX_{\frac{\pi}{12}}\}$.
    The nineteen regions are too many to name explicitly, but their hues indicate an increasing cost from a minimum at $\I$ to a maximum at $\SWAP$.
    There are also ten regions which have circuit depth at most two, hence they do not contribute volume and we suppress them from the picture.
    }
    \label{CXCX2CX3CoverageSets}
\end{figure}
Consider next the gate set $\S = \{\XX_{\frac{\pi}{4}}, \XX_x, \XX_y\}$ with the same cost function.
In \Cref{CostGraph2D}, we display the expected infidelity of synthesizing an $\S$--circuit for a Haar-randomly chosen unitary against both parameters $x$ and $y$.
The edges of this triangular figure degenerate to the case discussed in \Cref{Cost1DExample} along the lines $x = \frac{\pi}{4}$, $y = 0$, and $x = y$.
As before, the precise location of the optimum in the middle depends on the ratio $m / b$, but for experimentally realistic error models like the one depicted here, it is located near $(x, y) = (\frac{\pi}{8}, \frac{\pi}{12})$, this time achieving an expected infidelity of $1.51 \times 10^{-2}$.
Again, the basin is fairly wide and the minimum fairly independent of the value of $m / b$, so that $(\frac{\pi}{8}, \frac{\pi}{12})$ are good choices for inclusion in a native gate set even if the observed error model exhibits mild variation over time or across a device.
In \Cref{CXCX2CX3CoverageSets} we depict the optimal synthesis regions within the Weyl alcove for the gate set $\{\XX_{\frac{\pi}{4}}, \XX_{\frac{\pi}{8}}, \XX_{\frac{\pi}{12}}\}$.
\end{example}

\begin{example}
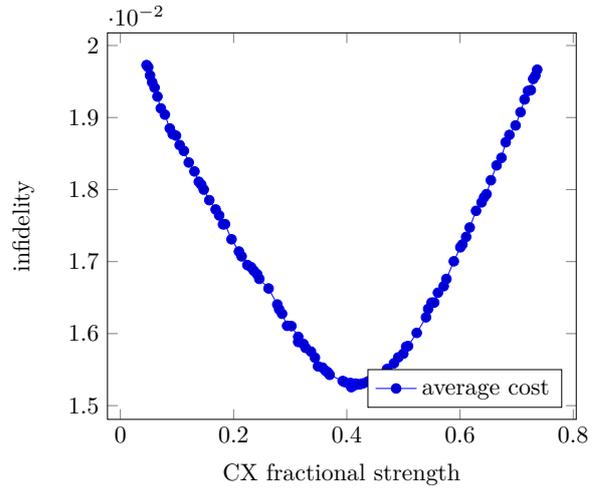
\begin{figure}[t]
    \begin{tikzpicture}[scale=0.90]
    \begin{axis}[
    xlabel={CX fractional strength},
    ylabel={infidelity},
    legend entries = {average cost},
    legend pos=south east,
    ]
    \addplot table[x expr={\thisrow{strength1}*3.1415/4},y=average_cost]{data/approx_gateset_landscape_1d.dat};
    \end{axis}
    \end{tikzpicture}
    \caption{
    The expected infidelity of a Haar-random operator decomposed approximately into $\S_x = \{\CX, \XX_x\}$.
    Uses an additive affine error model with offset $b \approx 1.909 \times 10^{-3}$ and slope $\frac{\pi}{4} \cdot m \approx 5.76 \times 10^{-3}$.
    }
    \label{Approx1DLandscape}
\end{figure}
\begin{figure}[t]
    \centering
    \begin{tikzpicture}
    \begin{axis}[
    yscale=2.0,
    xscale=0.80,
    xbar,
    ytick=data,
    yticklabels from table={data/approximate_circuit_dist.dat}{label},
    bar width=4pt,
    xmajorgrids=true,
    ]
    \pgfplotsset{every tick label/.append style={font=\tiny}}
    \addplot table[y=index, x=percent=]{data/approximate_circuit_dist.dat};
    \end{axis}
    \end{tikzpicture}
    \caption{%
    The distribution of circuit types encountered when approximately synthesizing 31,800 Haar-random two-qubit unitaries according to the additive affine error model with offset $b \approx 1.909 \times 10^{-3}$ and slope $\frac{\pi}{4} \cdot m \approx 5.76 \times 10^{-3}$.
    }
    \label{ApproximateCircuitDistribution}
\end{figure}
Taking these results for exact synthesis as inspiration, we can also explore effects introduced by approximate synthesis.
Our results here cannot be so clean, because we lose access to our method for analytic calculation, but we can still perform Monte Carlo experiments to analyze the relationship between $\S_x = \{\CX, \XX_x\}$ and the expected infidelity.
The plot in \Cref{Approx1DLandscape} shares many of the same qualitative features as \Cref{CostGraph1D} (e.g., the approximate position of the global minimum, and the non-concave kink near $x = \frac{\pi}{2} \cdot 1/3$), with an overall vertical shift coming from the approximation savings.
The global optimum for approximate synthesis into $\S_{x, y} = \{\CX, \XX_x, \XX_y\}$ is again near to the global optimum for exact synthesis, so we re-use the user-friendly value of $(x, y) = (\frac{\pi}{8}, \frac{\pi}{12})$ and depict in \Cref{ApproximateCircuitDistribution} the frequencies that these regions are used by approximate synthesis of Haar-random operations.
%
\end{example}

\begin{remark}\label{EfficiencyLowerBound}
\begin{figure}[t]
    \centering
    \resizebox{\columnwidth}{!}{%
    \begin{tabular}{@{}lccrr@{}} \toprule
    Gateset $\S$ & Approx.? & $\operatorname{argmin} / \frac{\pi}{4}$ & $\min \<\C_{\S}\>$ & (\%) \\
    \midrule
    $\{\CX\}$ & & - & $2.279 \times 10^{-2}$ & (100\%) \\
    $\{\CX\}$ & \checkmark & - & $2.058 \times 10^{-2}$ & (90.3\%) \\
    $\{\CX, \XX_x\}$ & & 0.52041 & $1.617 \times 10^{-2}$ & (70.9\%) \\
    $\{\CX, \XX_x\}$ & \checkmark & 0.51932 & $1.526 \times 10^{-2}$ & (67.0\%) \\
    $\{\CX, \XX_x, \XX_y\}$ & & $\begin{array}{c} 0.61856, \\ 0.41872 \end{array}$ & $1.510 \times 10^{-2}$ & (66.3\%) \\
    $\{\CX, \XX_x, \XX_y\}$ & \checkmark & $\begin{array}{c} 0.66362, \\ 0.43004 \end{array}$ & $1.445 \times 10^{-2}$ & (63.4\%) \\
    $\{\XX_{\mathit{cts}}\}$ & & - & $1.437 \times 10^{-2}$ & (63.1\%) \\
    $\{\XX_{\mathit{cts}}\}$ & \checkmark & - & $1.394 \times 10^{-2}$ & (61.2\%) \\
    \bottomrule
    \end{tabular}%
    }
    \caption{%
    The expected infidelity of a Haar-random operator with optimal decomposition into optimally chosen gate-sets of various sizes.
    Uses an additive affine error model with offset $b \approx 1.909 \times 10^{-3}$ and slope $\frac{\pi}{4} \cdot m \approx 5.76 \times 10^{-3}$.
    Note that $\{\XX_{\mathit{cts}}\}$ may be physically unrealistic.
    }
    \label{OptimaTable}
\end{figure}
In the limit where $\S$ contains all $\XX$ interactions, the most efficient circuit for $\CAN(a_1, a_2, a_3)$ is given by a product \[\CAN(a_1, a_2, a_3) = \]\[\CAN(a_1) \cdot \CAN(0, a_2) \cdot \CAN(0, 0, a_3),\] where each factor in the product is a Weyl reflection of a single $\XX$ gate of the same parameter, and where factors are dropped when the relevant parameter vanishes.
Under the assumption of an additive affine error model, this establishes a lower bound for how efficient we can expect our circuits to possibly be, as they are assembled from a more restrictive gate set.

Comparing \Cref{Cost1DExample} and \Cref{Cost2DExample}, we observe that there are rapidly diminishing returns to enlarging the native gate set.
Specialized to the same error model as in the Examples, the performance lower bound argued above is $(3/2 \cdot \frac{\pi}{4}) \cdot m + 3 b$, resulting in the table in \Cref{OptimaTable}.
\end{remark}


\begin{remark}
\begin{figure}[t]
    \centering
    \begin{tikzpicture}
    \begin{axis}[
    yscale=0.8,
    xscale=0.88,
    xbar,
    ytick=data,
    yticklabels from table={data/approximate_mirror_circuit_dist.dat}{label},
    bar width=4pt,
    xmajorgrids=true,
    ]
    \pgfplotsset{every tick label/.append style={font=\tiny}}
    \addplot table[y=index, x=percent=]{data/approximate_mirror_circuit_dist.dat};
    \end{axis}
    \end{tikzpicture}
    \caption{%
    The distribution of circuit types encountered when approximately synthesizing 49,700 Haar-random two-qubit unitaries, allowing mirroring, according to the additive affine error model with offset $b \approx 1.909 \times 10^{-3}$ and slope $\frac{\pi}{4} \cdot m \approx 5.76 \times 10^{-3}$.
    }
    \label{ApproximateMirrorCircuitDistribution}
\end{figure}
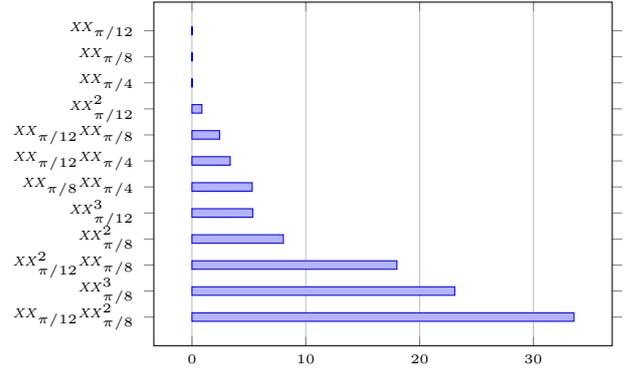
\begin{figure}[t]
    \centering
    \resizebox{\columnwidth}{!}{%
    \begin{tabular}{@{}lccrr@{}} \toprule
    Gateset $\S$ & Approx.? & $\operatorname{argmin} / \frac{\pi}{4}$ & $\min \<\C_{\S}\>$ & (\%) \\
    \midrule
    $\{\CX\}$ & & - & $2.279 \times 10^{-2}$ & (100\%) \\
    $\{\CX\}$ & \checkmark & - & $1.895 \times 10^{-2}$ & (83.2\%) \\
    $\{\CX, \XX_x\}$ & & $0.49970$ & $1.437 \times 10^{-2}$ & (63.1\%) \\
    $\{\CX, \XX_x\}$ & \checkmark & $0.46411$ & $1.352 \times 10^{-2}$ & (59.3\%) \\
    $\{\CX, \XX_x, \XX_y\}$ & & $\begin{array}{c} 0.49904, \\ 0.24991 \end{array}$ & $1.350 \times 10^{-2}$ & (59.2\%) \\
    $\{\XX_{\mathit{cts}}\}$ & & - & $1.304 \times 10^{-2}$ & (57.2\%) \\
    $\{\CX, \XX_x, \XX_y\}$ & \checkmark & $\begin{array}{c} 0.54244, \\ 0.37083 \end{array}$ & $1.300 \times 10^{-2}$ & (57.0\%) \\
    $\{\XX_{\mathit{cts}}\}$ & \checkmark & - & $1.241 \times 10^{-2}$ & (54.7\%) \\
    \bottomrule
    \end{tabular}%
    }
    \caption{%
    The expected infidelity of a Haar-random operator \emph{or its mirror} with optimal decomposition into optimally chosen gate-sets of various sizes.
    Uses an additive affine error model with offset $b \approx 1.909 \times 10^{-3}$ and slope $\frac{\pi}{4} \cdot m \approx 5.76 \times 10^{-3}$.
    Note that $\{\XX_{\mathit{cts}}\}$ may be physically unrealistic.
    }
    \label{MirrorTable}
\end{figure}
For a two-qubit unitary $U$, its \textit{mirror} is the gate $U \cdot \SWAP$.
The mirror of a canonical gate $\CAN(a_1, a_2, a_3)$ is again canonical, given by the formula \[\begin{cases}\CAN\left(\frac{\pi}{4} + a_3, \frac{\pi}{4} - a_2, \frac{\pi}{4} - a_1\right) & \text{when $a_1 \le \frac{\pi}{4}$}, \\ \CAN\left(\frac{\pi}{4} - a_3, \frac{\pi}{4} - a_2, a_1 - \frac{\pi}{4} \right) & \text{otherwise}. \end{cases}\]
This formula shows that mirroring interchanges the regions of $\A_{C_2}$ with the most and least infidelity cost, suggesting that our technique may be particularly fruitful at reducing the cost of mirrorable gates.
We summarize the numerical results in \Cref{MirrorTable}, and we depict in \Cref{ApproximateMirrorCircuitDistribution} the relative frequency of different circuit templates when synthesizing up to mirroring.
The main points are that these two synthesis strategies are ``compatible'', in that mirroring can be used in tandem with fractional synthesis to effect a combined decrease in infidelity, and that the optimal choice of finite gateset extension is somewhat different between the two but not wildly so.
\end{remark}

\begin{example}
\begin{figure}[t]
    \centering
    \resizebox{\columnwidth}{!}{%
    \begin{tabular}{@{}ccccr@{}} \toprule
    Gateset $\S$ & Approx.? & Mirror? & $\<\C_{\S}\>$ & (\%) \\
    \midrule
    $\{\CX, \XX_{\frac{\pi}{8}}\}$ & & & $1.619 \times 10^{-2}$ & (71.0\%) \\
    " & \checkmark & & $1.534 \times 10^{-2}$ & (67.3\%) \\
    " & & \checkmark & $1.437 \times 10^{-2}$ & (63.1\%) \\
    " & \checkmark & \checkmark & $1.374 \times 10^{-2}$ & (60.3\%) \\
    \midrule
    $\{\CX, \XX_{\frac{\pi}{8}}, \XX_{\frac{\pi}{12}}\}$ & & & $1.564 \times 10^{-2}$ & (68.6\%) \\
    " & \checkmark & & $1.483 \times 10^{-2}$ & (65.1\%) \\
    " & & \checkmark & $1.352 \times 10^{-2}$ & (59.3\%) \\
    " & \checkmark & \checkmark & $1.304 \times 10^{-2}$ & (57.2\%) \\
    \bottomrule
    \end{tabular}%
    }
    \caption{%
    The expected infidelity of a Haar-random operator with various optimal synthesis methods into a fixed pair of convenient gate sets.
    Uses an additive affine error model with offset $b \approx 1.909 \times 10^{-3}$ and slope $\frac{\pi}{4} \cdot m \approx 5.76 \times 10^{-3}$.
    Compare the expected fidelities with those advertised in \Cref{OptimaTable} and \Cref{MirrorTable}.
    }
    \label{NiceAnglesTable}
\end{figure}
We summarize the statistics on exact, approximate, and mirrored synthesis for the gate sets $\{\XX_{\frac{\pi}{4}}\}$, $\{\XX_{\frac{\pi}{4}}, \XX_{\frac{\pi}{8}}\}$, and $\{\XX_{\frac{\pi}{4}}, \XX_{\frac{\pi}{8}}, \XX_{\frac{\pi}{12}}\}$ in \Cref{NiceAnglesTable}.
These fractional gates are chosen because they are near to the unconstrained optima and because they are integer fractional iterates of $\CX$, which means they can be easily calibrated and benchmarked.
These are to be compared with the precise optima reported in \Cref{OptimaTable} and \Cref{MirrorTable}.
The main point is that using these fractional gates, rather than the optimal choice of fractional exponents, comes at only a modest cost in infidelity while greatly simplifying engineering.
\end{example}

\section*{Acknowledgements}

We would like to extend our thanks
to David Avis for assisting with \lrs~\cite{Avis,AvisFukuda};
to Jay Gambetta, for arranging the productive environment that gave rise to this work;
to Isaac Lauer, for keeping us grounded in the realities of working with superconducting hardware;
to Douglas Pearson, for helpful comments on a draft of this manuscript;
to Michael Mann, for providing the musical score to which this project was carried out;
and to the anonymous referees, whose helpful comments significantly improved the readability of this paper.

\clearpage

\bibliographystyle{quantum}
\bibliography{xx_synthesis}

\clearpage

\appendix
\section{Case-work for the approximation theorem}\label{ProofOfApproximationTheorem}

In this Appendix, we chase out the requisite case work to prove \Cref{ApproximationTheorem}.
We begin with some reductions.

\begin{lemma}
Let $a = (a_1, a_2, a_3)$ be a positive canonical triple, and let $P \subseteq \A_{C_2}$ be a polyhedron satisfying the reflection-closure property \[P = \left\{\left(\frac{\pi}{2} - b_1, b_2, b_3\right) \middle| b \in P\right\}.\]
The point $b \in P$ nearest in infidelity distance to the point $a \in \A_{C_2}$ satisfies $b_1 \le \frac{\pi}{4}$ if $a_1 \le \frac{\pi}{4}$ and $b_1 \ge \frac{\pi}{4}$ if $a_1 \ge \frac{\pi}{4}$.
\end{lemma}
\begin{proof}
In the expression \[\prod_j \cos^2(a_j - b_j) + \prod_j \sin^2(a_j - b_j),\] the bounds $0 \le a_2, b_2, a_3, b_3 \le \frac{\pi}{4}$ entail that the first summand is bounded from below by $\frac{1}{4} \cos^2(a_1 - b_1)$ and the second summand is bounded from above by $\frac{1}{4} \sin^2(a_1 - b_1)$.
Notice that replacing $b$ with its reflection $(\frac{\pi}{2} - b_1, b_2, b_3)$ trades the positions in the expression of $\cos^2(a_1 - b_1)$ and $\sin^2(a_1 - b_1)$.
Additionally, notice that the expression is maximized when $\cos^2(a_1 - b_1)$ takes on the larger of the two values, i.e., when $\cos^2(a_1 - b_1) \ge \frac{1}{2}$.
This then holds exactly when the conclusion of the Lemma does.
\end{proof}


\begin{corollary}
We need only handle the case $a_1 \le \frac{\pi}{4}$.
\end{corollary}
\begin{proof}
\Cref{MainGlobalTheorem} shows that $\XX$--circuit polytopes are reflection-invariant, so that we may apply the Lemma.
Furthermore, since average infidelity is invariant under replacing \emph{both} coordinates by their reflections
\begin{align*}
    (a_1, a_2, a_3) & \mapsto \left(\frac{\pi}{2} - a_1, a_2, a_3\right), \\
    (b_1, b_2, b_3) & \mapsto \left(\frac{\pi}{2} - b_1, b_2, b_3\right).
\end{align*}
we may reduce to one case of the Lemma, and we choose the case indicated in the Corollary statement.
\end{proof}

From here, we turn to the actual casework, walking first over the codimensions of the various facets of some implicitly understood $\XX$--circuit polytope and then over their possible slopes (noting that in each codimension there is a finite set of possibilities).
The codimension $0$ and $3$ (i.e., top- and bottom-dimensional) cases are trivial:

\begin{lemma}[{cf.\ \Cref{InfidelityProperties}}]
The infidelity functional $I|_a$ is extremized on the interior of a codimension $0$ facet if and only if $a$ is a member of that facet.
\qed
\end{lemma}

\begin{lemma}
The codimension $3$ facets contribute a finite set of points at which the restricted infidelity function $I|_a$ may be extremized.
\qed
\end{lemma}

\begin{lemma}
The infidelity functional $I|_a$ is extremized on the interior of codimension $1$ facets coincident with the outer walls of the Weyl alcove if and only if $a$ is a member of that facet.
\end{lemma}
\begin{proof}
We employ a strategy similar to \Cref{DisorderedMainTheorem}: our choice to restriction attention the alcove $\A_{C_2}$ is artificial, and it is equivalent to optimize the function $\min_{v, w \in W} I(v \cdot a, w \cdot b)$ where $W$ denotes the group of Weyl reflections, where the domain of $I$ is suitably extended by reflection beyond $\A_{C_2}$, and where $b^w$ is constrained to reside in $\bigcup_{w \in W} w \cdot P$, the closure of $P$ under Weyl reflections.
This closure is again a (possibly non-convex, possibly disconnected) polyhedron, but now the points $b \in P$ incident on the outer alcove walls belong to the interior of a codimension $0$ facet of $\bigcup_{w \in W} w \cdot P$.
Hence, the optimization condition reduces to that of the codimension $0$ facet case.
\end{proof}

\begin{lemma}
The infidelity functional $I|_a$ is extremized on the interior of the codimension $1$ facets not coincident with the outer walls of the Weyl alcove exactly at the nearest point in Euclidean distance.
\end{lemma}
\begin{proof}
Each such facet has an associated Lagrange multipliers problem, which we solve in turn.
We use the following abbreviations throughout:
\begin{align*}
    \partial_j I & := \frac{\partial I|_a}{\partial b_j}, &
    \delta_j & := a_j - b_j.
\end{align*}
The linear constraints on $a, b \in \A_{C_2}$ describe the following constraints on $\delta$:
\begin{align*}
    \delta_1 & \in \left[-\frac{\pi}{2}, \frac{\pi}{2}\right], &
    \delta_2 & \in \left[-\frac{\pi}{4}, \frac{\pi}{4}\right], &
    \delta_3 & \in \left[-\frac{\pi}{4}, \frac{\pi}{4}\right]. &
\end{align*}
Referring to \Cref{MainGlobalTheorem}, we break into cases based on the normal vector of the facet:
\begin{description}
    \item[$(0, 0, 1)$:]
	The Lagrange multiplier constraints are $\partial_1 I = 0$ and $\partial_2 I = 0$, which amount to the trigonometric conditions
	\begin{align*}
	(c_{2\delta_2} + c_{2\delta_3}) s_{2\delta_1} & = 0, &
	(c_{2\delta_1} + c_{2\delta_3}) s_{2\delta_2} & = 0.
	\end{align*}
	Taking into account the domain constraints on $\delta$, we see that the first equation's is satisfied either when $a$ and $b$ both represent the identity unitary or when $\delta_1 = 0$.
	Taking into account the deduced constraint $\delta_1 = 0$, the second clause is then satisfied only when $\delta_2 = 0$.
	Finally, $b_3$ is determined by being constrained to the frustrum plane.%
	\footnote{A special case of this was previously investigated by Cross et al.~\cite[Equation B8f]{CBSNG} when finding the best approximant using a pair of $\CX$ gates.}
    \item[$(1, 1, 1)$:]
    The Lagrange multiplier constraints are $\partial_1 I = \partial_2 I$ and $\partial_2 I = \partial_3 I$, which amount to the trigonometric conditions
    \begin{align*}
        (c_{\delta_1 - \delta_2} + c_{\delta_1 + \delta_2} c_{2 \delta_3}) s_{\delta_1 - \delta_2} & = 0, \\
        (c_{\delta_2 - \delta_3} + c_{\delta_2 + \delta_3} c_{2 \delta_1}) s_{\delta_2 - \delta_3} & = 0.
    \end{align*}
    These equalities are analyzed similarly to as in the previous case.
    The first equation is satisfied either when both $a$ and $b$ represent the identity unitary or when $\delta_1 = \delta_2$, and the second equality is similarly dispatched to give $\delta_2 = \delta_3$.
    \item[$(-1, 1, 1)$:]
    The Lagrange multiplier constraints are $-\partial_1 I = \partial_2 I$ and $-\partial_1 = \partial_3 I$, which amount to the trigonometric conditions
    \begin{align*}
        (c_{\delta_1 + \delta_2} + c_{\delta_1 - \delta_2} c_{2 \delta_3}) s_{\delta_1 + \delta_2} & = 0, \\
        (c_{\delta_1 + \delta_3} + c_{\delta_1 - \delta_3} c_{2 \delta_2}) s_{\delta_1 + \delta_3} & = 0.
    \end{align*}
    Reasoning identically about the domains, we conclude that $-\delta_1 = \delta_2$ and $-\delta_1 = \delta_3$.
\end{description}
In each case, the critical points are seen to lie at the Euclidean projections onto the relevant planes.
\end{proof}

\begin{lemma}\label{InnerWallLineTheorem}
The infidelity functional $I|_a$ is extremized on the interior of the codimension $2$ facets \emph{not coincident} with the outer walls of the Weyl alcove exactly at the nearest point in Euclidean distance.
\end{lemma}
\begin{proof}
Again, we are tasked with solving a family of constrained optimization problems.
This time, each nondegenerate pair of inner walls intersect at a line with tangent vector $v$, and we are looking to solve along the line for the condition $\nabla I|_a \cdot v = 0$.
To parameterize the line, we select a vertex $b \in \A_{C_2}$ on it and set \[\ell(t) = v \cdot t + b.\]
We break $v$ (i.e., the choice of plane pair) into cases.
\begin{description}
    \item[$(-1, -1, 0)$:]
    This tangent vector $v$ arises from the intersection of planes with normal vectors $(0, 0, 1)$ and $(-1, 1, 1)$.
    Expanding $(\nabla I|_{a_1, a_2, a_3} \cdot v)$ yields \[(c_{\delta_1 - \delta_2 + 2 \delta_3} + c_{\delta_1 - \delta_2 - 2\delta_3} + 2 c_{\delta_1 + \delta_2 + 2t}) s_{\delta_1 + \delta_2 + 2t} = 0,\] where $t$ is constrained to \[b_1 - \frac{\pi}{4} \le t \le b_2 - b_3.\]
    \item[$(1, -1, 0)$:]
    This tangent vector $v$ arises from the intersection of planes with normal vectors $(0, 0, 1)$ and $(1, 1, 1)$.
    Expanding $(\nabla I|_{a_1, a_2, a_3} \cdot v)$ yields \[(c_{\delta_1 + \delta_2 + 2\delta_3} + c_{\delta_1 + \delta_2 - 2\delta_3} + 2 c_{\delta_1 - \delta_2 - 2t}) s_{\delta_1 - \delta_2 - 2t} = 0,\] where $t$ is constrained to \[\frac{1}{2}(b_2 - b_1) \le t \le \min\left\{b_2 - b_3, \frac{\pi}{4} - b_1\right\}.\]
    \item[$(0, 1, -1)$:]
    This tangent vector $v$ arises from the intersection of planes of normal vectors $(1, 1, 1)$ and $(-1, 1, 1)$.
    Expanding $(\nabla I|_{a_1, a_2, a_3} \cdot v)$ yields \[(c_{2\delta_1 + \delta_2 + \delta_3} + c_{2\delta_1 - \delta_2 - \delta_3} + 2 c_{\delta_2 - \delta_3 - 2t}) s_{\delta_2 - \delta_3 - 2t} = 0,\] where $t$ is constrained to \[\frac{1}{2}(b_3 - b_2) \le t \le \min\left\{ b_1 - b_2, b_3 \right\}.\]
\end{description}
In each case, the first clause is not satisfiable on the indicated interval, and the second clause contributes at most only the Euclidean critical point.
\end{proof}

\begin{lemma}
The infidelity functional $I|_a$ is extremized on the interior of the codimension $2$ facets \emph{coincident} with the outer walls of the Weyl alcove exactly at the nearest point in Euclidean distance.
\end{lemma}
\begin{proof}
As in \Cref{InnerWallLineTheorem}, we intend to split over the slopes of the plane-plane intersections.
Two of these cases are familiar: since the outer alcove wall $b_3 \ge 0$ shares a normal with the frustrum inequality of \Cref{MainGlobalTheorem}, the tangent vectors $(-1, -1, 0)$ and $(1, -1, 0)$ both reappear, and we have already dispatched them in the proof of \Cref{InnerWallLineTheorem}.
The frustrum inequality contributes one codimension $2$ facet not covered by the above: its intersection with the wall $a_2 \ge a_3$ yields a line with tangent vector $(1, 0, 0)$, and the associated optimization problem is \[(c_{2 \delta_2} + c_{2 \delta_3}) s_{2(\delta_1 - t)} = 0.\]
The sine factor contributes the Euclidean critical point, and the cosine factor is independent of $t$.

The remaining cases correspond to ``inner creases'' in the Weyl-closed solid $\cup_{w \in W} w \cdot P$, and they are treated quite differently.
In each case, the strategy is to show that the facet is irrelevant (i.e., has no critical points) unless the outer alcove inequality is tight for the point $a$, then to use that tightness to simplify the expression further.
Our strategy for showing irrelevance is to show that, when $a$ is not a member of an outer facet, $\nabla I|_a$ has a nonnegative inner product with the inward-facing normal of the codimension $2$ facet considered as part of the boundary of the inner codimension $1$ facet.
Taking this as given, we would learn that the extremum then would always lie on the codimension $1$ facet, so that we could avoid considering the codimension $2$ facet.
In fact, this strategy gives us a bit more: even \emph{without} the assumption that $a$ lies off of the outer wall, continuity would show that this conclusion still holds for extrema, since the assumption is only violated at limit points of open regions.%
\footnote{%
    Importantly, we are \emph{not} arguing about critical points but about extrema.
    Critical points can manifest on a boundary via a sequence of points on the bulk which themselves are merely \emph{approximately} critical points, without exactly being critical points.
    However, any such critical point cannot yield a more extreme value than the value achieved by the function on a sequence of values in the bulk which are extrema for the functional constrained to planes parallel to the outer facet.
}
Thus, we can avoid investigating even the aforementioned simplified expressions, leaving open only the task of exhibiting a positive inner product with the inward-facing normal.
\begin{description}
    \item[$(-2, 1, 1)$:]
    This tangent vector $v$ arises from the intersection of the inner wall with normal $n_i = (1, 1, 1)$ and outer wall with normal $n_o = (0, 1, -1)$.
    Assuming $a_2 > a_3$, we would like to show that the following quantity is positive: \[\nabla I|_a \cdot n_i = (c_{a_2 - a_3} + c_{a_2 + a_3 - 2t} c_{2(a_1 - b_1 + 2t)}) s_{a_2 - a_3},\] where $(b_1, 0, 0)$ lies on the line and $t$ satisfies $0 \le t \le \frac{1}{3} b_1$.
    \item[$(1, 1, -2)$:]
    This tangent vector $v$ arises from the intersection of the inner wall with normal $n_i = (1, 1, 1)$ and outer wall with normal $n_o = (1, -1, 0)$.
    Assuming $a_1 > a_2$, we would like to show that the following quantity is positive: \[\nabla I|_a \cdot n_i = (c_{a_1 - a_2} + c_{2(a_3 + 2t)} c_{a_1 + a_2 - 2 b_1 - 2t}) s_{a_1 - a_2},\] where $(b_1, b_1, 0)$ lies on the line and $t$ satisfies $-\frac{1}{3} b_1 \le t \le 0$.
    \item[$(-2, -1, -1)$:]
    This tangent vector $v$ arises from the intersection of the inner wall with normal $n_i = (-1, 1, 1)$ and outer wall with normal $n_o = (0, 1, -1)$.
    Assuming $a_2 > a_3$, we would like to show that the following quantity is positive \[\nabla I|_a \cdot n_i =\]\[(c_{a_2 - a_3} + c_{a_2 + a_3 - 2 b_1 + 2t} c_{2(a_1 - b_1 + 2t)}) s_{a_2 - a_3},\] where $(b_1, 0, 0)$ lies on the line and $t$ satisfies $\frac{1}{2} b_1 - \frac{\pi}{8} \le t \le 0$.
    %
\end{description}
In each case, the domain restrictions cause the arguments to sine and cosine to lie in the positive range.
\end{proof}

\section{Inclusion-exclusion and incidence degeneracy}\label{PIESection}

\begin{figure}
    \centering
    \begin{tikzpicture}
    \begin{axis}[
    xlabel={actual convex volume calculations},
    ylabel={naive convex volume calculations},
    ymode=log,
    ]
    \addplot[only marks] table {data/inclusion_exclusion.dat};
    \end{axis}
    \end{tikzpicture}
    \caption{%
    Volume computations made during an example gate set exploration exercise.
    The vertical coordinate shows the number of convex volume computations required with a naive application of the inclusion-exclusion formula, displayed on a logarithmic scale.
    The horizontal coordinate shows the number of convex volume computations actually performed when using the method described in \Cref{PIESection}.
    Points near the top-left indicate ``false complexity'' in the convex polytope arrangement, and points near the bottom-right indicate ``true complexity''.
    }
    \label{PIERuntimeScatterplot}
\end{figure}
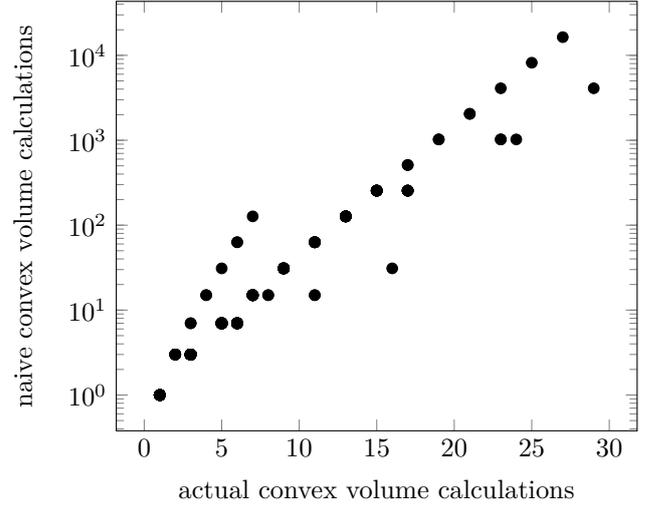

In uncovering our main results, it was invaluable to be able to calculate the volume of a nonconvex polytope.
Not only did volume calculations play an outsized role in \Cref{GatesetOptSection}, they also underlie primitive operations.
For instance, while containment of a polytope $P$ within a \emph{convex} polytope $Q$ can be checked on vertices, this is not true of two generic polytopes; instead, assuming that $P$ is of constant dimension, $P \subseteq Q$ if and only if $\vol(P) = \vol(P \cap Q)$.
For this reason, we found it imperative to have a robust and efficient method for volume calculation.

The process of volume calculation cleaves into two parts: reducing to the convex case, and computing the volume of convex components.
Both steps admit several approaches: for instance, the former can be accomplished by (joint) triangulation, and the latter can be accomplished by determinant methods.
However, it is difficult to come by implementations of these techniques which are open-source, permissively licensed, accurate / exact, and which operate in high dimension.%
\footnote{%
See \cite{AndersonHiller} for a notable exception.
}
In our setting, we can often get away with the following: for the second step, use the (somewhat computationally expensive) ability of a computer algebra system, such as \lrs, to calculate the volume of a single convex polytope; and for the first step, use a variant of inclusion-exclusion.

The naive application of inclusion-exclusion is described by
\begin{align*}
\vol\left( \bigcup_{j \in J} P_j \right) & = -\sum_{I \subseteq J} (-1)^{|I|} \vol\left( \bigcap_{i \in I} P_i \right) \\
& =: -\sum_{I \subseteq J} (-1)^{|I|} \vol\left( P_I \right).
\end{align*}
The terms on the right-hand side are all volumes of convex bodies, hence are individually approachable, but there are $2^{|J|}$ such summands.
These summands can be culled in two ways:

\begin{enumerate}
    \item
    Terms with vanishing volume are \textit{downward-closed}:
    If $\vol P_I = 0$, then $\vol P_{I' \cup I} = 0$ for any $I'$.
    \item
    Containment is \textit{downward-closed}:
    If $\vol P_I = \vol P_{j \cup I}$, then $\vol P_{I' \cup I} = \vol P_{I' \cup j \cup I}$ for any $I'$.
    For $j \not\in I \cup I'$, these pairs of values appear with opposite sign in the larger sum and cancel each other out.
\end{enumerate}

\noindent
It is simple to cull summands with the first observation: whenever we encounter a summand with vanishing volume, we can skip all of its descendants.
The second observation is trickier: after encountering two pairs $(j_1, I_1)$ and $(j_2, I_2)$ which fit the hypothesis, it is possible to double-count a term as belonging to two canceling pairs.

The following procedure accounts for this wrinkle.
We will maintain two ``skip lists'' of indices to ignore:
\begin{enumerate}
    \item
    A skip of Type 1 corresponds to an intersection which vanishes exactly, and it is recorded by a single bitmask of the entries which populate $I$.
    \item
    A skip of Type 2 corresponds to an intersection which cancels with one of its immediate descendants, and it is recorded by a bitmask of the entries which populate $I$ as well as the index $j$ of the descendant (which does not belong to $I$).
\end{enumerate}
We traverse the possible depths of intersections, and at each depth we traverse the possible intersections at that depth.
For each intersection, if it match either skip list, we ignore it and continue to the next intersection at this depth.
Otherwise, we compute the volume of this intersection.
If the volume vanishes, we add this index to the Type 1 skip list, then continue as if we have done no work at this step.
If the volume is equal to one of our immediate predecessors, we add to the Type 2 skip list its index and the extra intersection factor $j$ which witnesses us as its child, then continue as if we have done no work at this step.
Otherwise, we add the nonzero contribution to the running alternating sum with the appropriate sign.
When we exhaust the possible intersections at this depth, if we have performed no work, we terminate the iteration altogether; otherwise, we proceed to the next depth.

Now, we double back to reintroduce the summands which we previously double-counted, which we formulate in a way to also avoid double-counting the double-countings.
Traversing the Type 2 skip list in the order in which it was created, let us consider the $t$\th mask and toggle $(I_t, j_t)$, as well as some intermediate $s$\th mask and toggle $(I_s, j_s)$ with $s < t$ and with $j_t \in I_s$.
Double-counting occurs for this pair at an intersection $I$ when the following are met:
\begin{enumerate}
    \item
    The $t$\th mask matches $I_t \le I$.
    \item
    The $t$\th toggle is disabled: $j_t \not\in I$.
    \item
    The $s$\th mask matches after the toggle is enabled: $I_s \le I \cup \{j_t\}$.
    \item
    For all earlier $s' < s$, the $s'$\th mask does not include the $t$\th toggle and additionally does not match $I$.
    \item
    For all later $s < t' < t$, the $t'$\th mask does not match the toggle-on form $I \cup \{j_t\}$.
\end{enumerate}
Whenever these constraints are met, we reintroduce the summand at $I$ to the running alternating sum.
After iterating over all possible values of $s$ and $t$, the running sum is the true alternating sum.

For any $s < t$, the constraints on $I$ described above are quite strong (and often even contradictory), so that iterating over the possible ways to satisfy these constraints, rather than iterating over $I$ and checking satisfaction, frequently results in loops with few to no iterations.
In one instance ``in the wild'', this strategy reduced a calculation from $2^{14} - 1 \approx 16,000$ convex volume computations to a mere $27$ volume computations.


\onecolumn
\clearpage

\begin{figure}[p]
    \centering
    \begin{align*}
    \begin{array}{c}
    \text{$b$--coord's unreflected,} \\
    \text{$a_f = b_3$,} \\
    \text{slant inequality on $a_h$}
    \end{array} %
    & : \left\{ %
    \begin{array}{rcl}
    \alpha_+ + \beta & \ge & b_1 + b_2 + b_3, \\ 
    \alpha_+ - 2 \alpha' + \beta & \ge & -b_1 + b_2 + b_3, \\ 
    \alpha_+ - \beta & \ge & -b_1 + b_2 + b_3, \\ 
    \alpha_+ - \alpha' & \ge & b_3, \\ 
    \alpha_+ - \alpha' - \alpha'' + \beta & \ge & b_2.
    \end{array} \right. \\
    \begin{array}{c}
    \text{$b$--coord's reflected,} \\
    \text{$a_f = b_3$,} \\
    \text{strength inequality on $a_h$}
    \end{array}
    & : \left\{
    \begin{array}{rcl}
    \frac{\pi}{2} + \alpha_+ - 2 \alpha' + \beta & \ge & b_1 + b_2 + b_3, \\
    -\frac{\pi}{2} + \alpha_+ + \beta & \ge & -b_1 + b_2 + b_3, \\
    \alpha_+ - \beta & \ge & b_1 + b_2 + b_3, \\
    \alpha_+ - \alpha' & \ge & b_3, \\
    \alpha_+ - \alpha' - \alpha'' + \beta & \ge & b_2.
    \end{array} \right. \\
    \begin{array}{c}
    \text{$b$--coord's unreflected,} \\
    \text{$a_f = b_1$,} \\
    \text{slant inequality on $a_f$}
    \end{array}
    & : \left\{
    \begin{array}{rcl}
    \alpha_+ + -2 \alpha' + \beta & \ge & -b_1 + b_2 + b_3, \\
    \alpha_+ + \beta & \ge & b_1 + b_2 + b_3, \\
    \frac{\pi}{2} - \beta & \ge & b_1 - b_2 + b_3, \\
    \alpha_+ + -2 \alpha' - \beta & \ge & -b_1 - b_2 + b_3, \\
    \alpha_+ - \beta & \ge & b_1 - b_2 + b_3, \\
    \alpha_+ - \alpha' - \alpha'' + \beta & \ge & b_3.
    \end{array} \right. \\
    \begin{array}{c}
    \text{$b$--coord's reflected,} \\
    \text{$a_f = b_1$,} \\
    \text{strength inequality on $a_f$}
    \end{array}
    & : \left\{
    \begin{array}{rcl}
    -\frac{\pi}{2} + \alpha_+ + \beta & \ge & -b_1 + b_2 + b_3, \\
    \frac{\pi}{2} + \alpha_+ + -2 \alpha' + \beta & \ge & b_1 + b_2 + b_3, \\
    \frac{\pi}{2} - \beta & \ge & b_1 - b_2 + b_3, \\
    -\frac{\pi}{2} + \alpha_+ - \beta & \ge & -b_1 - b_2 + b_3, \\
    \frac{\pi}{2} + \alpha_+ + -2 \alpha' - \beta & \ge & b_1 - b_2 + b_3, \\
    \alpha_+ - \alpha' - \alpha'' + \beta & \ge & b_3.
    \end{array} \right.
    \end{align*}
    \caption{%
    Four tables of inequalities describing the four regions of $b$--coordinates from \Cref{RelevantConvexSummands}.
    See also \Cref{LiftedInequalityTables}.
    }
    \label{ProjectedInequalityTables}
\end{figure}

\begin{figure}[p]
    \centering
    \begin{align*}
    \begin{array}{c}
    \text{$b$--coord's unreflected,} \\
    \text{$a_f = b_3$}, \\
    \text{slant inequality on $a_h$}
    \end{array}
    & : \left\{
    \begin{array}{rcl}
    - b_3 + \alpha_+ & \ge & a_h + a_\ell, \\
    - b_3 + \alpha_+ - 2 \alpha' & \ge & -a_h + a_\ell, \\
    \alpha_+ - \alpha' - \alpha'' & \ge & a_\ell, \\
    b_1 + b_2 + \beta & \ge & a_h + a_\ell, \\
    \pi - b_1 - b_2 - \beta & \ge & a_h + a_\ell, \\
    - b_1 - b_2 + \beta & \ge & -a_h - a_\ell, \\
    b_1 - b_2 + \beta & \ge & a_h - a_\ell, \\
    - b_1 + b_2 + \beta & \ge & -a_h + a_\ell, \\
    b_1 - b_2 - \beta & \ge & -a_h + a_\ell.
    \end{array} \right. \\
    \begin{array}{c}
    \text{$b$--coord's reflected,} \\
    \text{$a_f = b_3$}, \\
    \text{strength inequality on $a_h$}
    \end{array}
    & : \left\{
    \begin{array}{rcl}
    \frac{\pi}{2} - b_3 + \alpha_+ - 2 \alpha' & \ge & a_h + a_\ell, \\
    -\frac{\pi}{2} - b_3 + \alpha_+ & \ge & -a_h + a_\ell, \\
    \alpha_+ - \alpha' - \alpha'' & \ge & a_\ell, \\
    b_1 + b_2 + \beta & \ge & a_h + a_\ell, \\
    \pi - b_1 - b_2 - \beta & \ge & a_h + a_\ell, \\
    - b_1 - b_2 + \beta & \ge & -a_h - a_\ell, \\
    b_1 - b_2 + \beta & \ge & a_h - a_\ell, \\
    - b_1 + b_2 + \beta & \ge & -a_h + a_\ell, \\
    b_1 - b_2 - \beta & \ge & -a_h + a_\ell.
    \end{array} \right. \\
    \begin{array}{c}
    \text{$b$--coord's unreflected,} \\
    \text{$a_f = b_1$}, \\
    \text{slant inequality on $a_f$}
    \end{array}
    & : \left\{
    \begin{array}{rcl}
    - b_1 + \alpha_+ & \ge & a_h + a_\ell, \\
    b_1 + \alpha_+ + -2 \alpha' & \ge & a_h + a_\ell, \\
    \alpha_+ - \alpha' - \alpha'' & \ge & a_\ell, \\
    b_2 + b_3 + \beta & \ge & a_h + a_\ell, \\
    \pi - b_2 - b_3 - \beta & \ge & a_h + a_\ell, \\
    - b_2 - b_3 + \beta & \ge & -a_h - a_\ell, \\
    b_2 - b_3 + \beta & \ge & a_h - a_\ell, \\
    - b_2 + b_3 + \beta & \ge & -a_h + a_\ell, \\
    b_2 - b_3 - \beta & \ge & -a_h + a_\ell.
    \end{array} \right. \\
    \begin{array}{c}
    \text{$b$--coord's reflected,} \\
    \text{$a_f = b_1$}, \\
    \text{strength inequality on $a_f$}
    \end{array}
    & : \left\{
    \begin{array}{rcl}
    \frac{\pi}{2} - b_1 + \alpha_+ + -2 \alpha' & \ge & a_h + a_\ell, \\
    -\frac{\pi}{2} + b_1 + \alpha_+ & \ge & a_h + a_\ell, \\
    \alpha_+ - \alpha' - \alpha'' & \ge & a_\ell, \\
    b_2 + b_3 + \beta & \ge & a_h + a_\ell, \\
    \pi - b_2 - b_3 - \beta & \ge & a_h + a_\ell, \\
    - b_2 - b_3 + \beta & \ge & -a_h - a_\ell, \\
    b_2 - b_3 + \beta & \ge & a_h - a_\ell, \\
    - b_2 + b_3 + \beta & \ge & -a_h + a_\ell, \\
    b_2 - b_3 - \beta & \ge & -a_h + a_\ell.
    \end{array} \right.
    \end{align*}
    \caption{%
    Four tables of inequalities describing the relationship between the $a$--coordinates and the $b$--coordinates in the four regions of \Cref{RelevantConvexSummands}.
    See also \Cref{ProjectedInequalityTables}.
    }
    \label{LiftedInequalityTables}
\end{figure}

\end{document}